\RequirePackage[l2tabu, orthodox]{nag}		
\documentclass[10pt]{article}

\usepackage[T1]{fontenc}
\usepackage{lmodern}
\usepackage[frenchmath]{mathastext}							
\usepackage{natbib}    
\setcitestyle{authoryear}
\usepackage{subcaption}
\usepackage{amsmath}                        
\numberwithin{equation}{section}
\usepackage{graphicx} 											
\usepackage{enumitem}												
\usepackage{mdwlist}												
\usepackage[dvipsnames]{xcolor}							
\usepackage[plainpages=false, pdfpagelabels]{hyperref} 
	\hypersetup{
		colorlinks   = true,
		citecolor    = RoyalBlue, 
		linkcolor    = RubineRed, 
		urlcolor     = Turquoise
	}
\usepackage{amsmath}
\usepackage{amssymb}                        
\usepackage[mathscr]{eucal}                 
\usepackage{dsfont}													
\usepackage[paperwidth=8.5in,paperheight=11in,top=1.25in, bottom=1.25in, left=1.0in, right=1.0in]{geometry} 
\usepackage{mathtools}                      
\mathtoolsset{showonlyrefs=true}          
\usepackage{amsthm}                         
	\allowdisplaybreaks                       
	\theoremstyle{plain}
	\newtheorem{theorem}{Theorem}
	\numberwithin{theorem}{section}
	\newtheorem{lemma}[theorem]{Lemma}       	
	\newtheorem{proposition}[theorem]{Proposition}
	
	\theoremstyle{definition}
	
	\newtheorem{example}[theorem]{Example}

\renewcommand{\[}{\left[}

\newcommand\Cb{\mathds{C}}
\newcommand\Eb{\mathds{E}}
\newcommand\Fb{\mathds{F}}

\newcommand\Pb{\mathds{P}}

\newcommand\Rb{\mathds{R}}

\newcommand\Ac{\mathscr{A}}

\newcommand\Fc{\mathscr{F}}
\newcommand\Gc{\mathscr{G}}

\newcommand\Lc{\mathscr{L}}

\newcommand\Oc{\mathscr{O}}
\newcommand\Pc{\mathscr{P}}

\newcommand\Zc{\mathscr{Z}}

\newcommand\eps{\varepsilon}
\newcommand\om{\omega}
\newcommand\Om{\Omega}
\newcommand\sig{\sigma}
\newcommand\Sig{\Sigma}
\newcommand\Lam{\Lambda}
\newcommand\gam{\gamma}
\newcommand\Gam{\Gamma}
\newcommand\lam{\lambda}
\newcommand\del{\delta}

\newcommand\zb{\bar{z}}

\newcommand\Tb{{\overline{T}}}

\newcommand\rhob{\bar{\rho}}

\newcommand{\vb}{\bar{v}}
\newcommand{\Sigb}{\bar{\Sig}}
\newcommand\alphab{\bar{\alpha}}
\newcommand\betab{\bar{\beta}}
\newcommand\gamb{\bar{\gam}}

\newcommand\lamb{\bar{\lam}}
\newcommand\Psib{\bar{\Psi}}
\newcommand\Phib{\bar{\Phi}}
\newcommand\zetab{\bar{\zeta}}

\newcommand\Cv{\mathbf{C}}
\newcommand\mv{\mathbf{m}}

\newcommand\Ebt{\widetilde{\Eb}}
\newcommand\Pbt{\widetilde{\Pb}}
\newcommand\Act{\widetilde{\Ac}}

\newcommand\Yt{\widetilde{Y}}

\newcommand\Wt{\widetilde{W}}

\newcommand\yt{\widetilde{y}}

\newcommand\sigt{\widetilde{\sig}}
\newcommand\mut{\widetilde{\mu}}

\renewcommand\d{\partial}

\newcommand\ii{\mathtt{i}}
\newcommand\dd{\mathrm{d}}
\newcommand\ee{\mathrm{e}}

\renewcommand\phi{\varphi}
\newcommand\phih{\widehat{\phi}}

\newcommand{\BS}{\text{BS}}

\begin{document}

\title{Options on Bonds: Implied Volatilities from Affine Short-Rate Dynamics}

\author{
Matthew Lorig
\thanks{Department of Applied Mathematics, University of Washington.  \textbf{e-mail}: \url{mlorig@uw.edu}}
\and
Natchanon Suaysom
\thanks{Department of Applied Mathematics, University of Washington.  \textbf{e-mail}: \url{nsuaysom@uw.edu}}
}

\date{This version: \today}

\maketitle


\begin{abstract}
We derive an explicit asymptotic approximation for the implied volatilities of Call options written on bonds assuming the short-rate is described by an affine short-rate model.  For specific affine short-rate models, we perform numerical experiments in order to gauge the accuracy of our approximation.  
\end{abstract}

%
%

\section{Introduction}
\label{sec:intro}
 
\textit{Affine short-rate models} refer to a class of interest rate models in which the price of any zero-coupon bond can be expressed as the exponential of affine function of the instantaneous short-rate.  Well-known affine short-rate models include the Vasicek \cite{vasicek1977equilibrium}, Cox-Ingersoll-Ross {(CIR)} \cite{cox2005theory}, Hull-White \cite{hull1990pricing} and Fong-Vasicek \cite{fong1991fixed} models, as well as their multi-factor versions.  Such models enjoy wide popularity among practitioners and academics alike because these models are flexible enough to fit the observed yield curve and easy to calibrate, due to the closed-form expression for bond prices, and hence yields.
\\[0.5em]
Despite their widespread use in yield-curve modeling, affine short-rate models are rarely used to price options on bonds or calibrate to the implied volatility surface of bond options.  For this task, practitioners assume forward prices of bonds are modeled by a local-stochastic volatility (LSV) model.  In particular, the SABR model \cite{hagan2002managing}, is often used as a model for forward bond prices because it admits an explicit approximation of implied volatility, which can be used {to calibrate} to observed implied volatilities. 
\\[0.5em]
Yet, if one assumes an affine model for the {short-rate}, the resulting forward bond prices will not have SABR dynamics.  As a result, if a bank uses an affine short-rate model to describe the yield curve, and the SABR model to describe the implied volatility surface of options on bonds, the bank is using two different models for the short-rate.  Such a practice clearly introduces arbitrage into the market.
\\[0.5em]
The purpose of this paper is to derive an explicit approximation for the implied volatilities of options on bonds assuming the short-rate is of the affine class.  In doing so, we provide a \textit{unified framework} for calibrating both to observed yields and to observed implied volatilities.  To derive the implied volatility approximation, we use the polynomial expansion method that was introduced by \cite{pagliarani2011analytical} in order to derive approximate prices for options on equity in a scalar setting and {later} extended in \cite{lorig-pagliarani-pascucci-2} in order to obtain approximate implied volatilities in a multi-factor LSV setting. 
\\[0.5em]
The rest of this paper proceeds as follows:
in Section \ref{sec:model} we introduce the class of affine short-rate models that we will consider in this paper and in Section \ref{sec:review} we briefly review how one can compute prices for bonds and options on bonds in the affine short-rate setting.  In Section \ref{sec:lsv} we provide an explicit relation between affine short-rate models and classical local-stochastic volatility models.  We use this relation in Sections \ref{sec:price-asymptotics} and \ref{sec:imp-vol} to develop explicit approximations for the prices of options on bonds and their corresponding implied volatilities.  In Section \ref{sec:examples}, we perform a number of numerical experiments to gauge the accuracy of our implied volatility approximation in four specific affine term-structure models: Vasicek, CIR, two-dimensional CIR and Fong-Vasicek.  Some thoughts on future work are offered in Section \ref{sec:conclusion}.

%
%

\section{Model and Assumptions}
\label{sec:model}
Throughout this paper, we will consider a financial market over a time horizon from zero to $\Tb < \infty$ with no arbitrage and no transactions costs.
As a starting point, we fix a complete probability space $(\Om,\Fc,\Pb)$ and a filtration $\Fb = (\Fc_t)_{0 \leq t \leq \Tb}$.  The probability measure $\Pb$ represents the market's chosen pricing measure taking the \textit{money market account} $M = (M_t)_{0 \leq t \leq \Tb}$ as num\'eraire.
The filtration $\Fb$ represents the history of the market.
\\[0.5em]
We shall assume that money market account $M$ has dynamics of the form
\begin{align}
\dd M_t
	&=	R_t M_t \dd t ,  \label{eq:dM}
\end{align}
where $R = (R_t)_{t \geq 0}$ is the instantaneous {\textit{short-rate}} of interest.  We further suppose that the {short-rate} $R$ is given by
\begin{align}
R_t
	&=	r(Y_t) , \label{eq:R_t}
\end{align}
for some function $r : \Rb^d \to \Rb_+$ and some Markov diffusion process $Y = (Y_t^{(1)}, Y_t^{(2)}, \ldots, Y_t^{(d)})$.  Specifically, we suppose that $Y$ is the unique strong solution of a stochastic differential equation (SDE) of the form
\begin{align}
\dd Y_t
	&=	\mu(t,Y_t) \dd t + \sig(t,Y_t) \dd W_t , \label{eq:Y_t} 
\end{align}
for some functions $\mu : [0,\Tb] \times \Rb^d \to \Rb^d$ and $\sig : [0,\Tb] \times \Rb^d \to \Rb^{d \times d}$, where $W = (W_t^{(1)}, W_t^{(2)}, \ldots, W_t^{(d)})_{t \geq 0}$ is a $d$-dimensional $(\Pb,\Fb)$-Brownian motion.  Thus, the $i$th component of $Y$ is given by
\begin{align}
\dd Y_t^{(i)}
	&=	\mu_i(t,Y_t) \dd t + \sum_{j=1}^d \sig_{i,j}(t,Y_t) \dd W_t^{(j)} . \label{eq:dYi_t}
\end{align}
Lastly, we shall assume that $R$ is an \textit{affine {short-rate} model}, meaning that the functions $(r,\mu,\sig)$ satisfy
\begin{align}
r(y)
	&=	q + \sum_{i=1}^d {\psi_i} y_i , &
\mu(t,y)  
	&= b(t) + \sum_{i=1}^d \beta_i(t) y_i , &
\sigma(t,y)\sigma^\text{Tr}(t,y) 
	&= \ell(t) + \sum_{i=1}^d \lambda_i(t) y_i ,  \label{eq:r-mu-sigma}
\end{align}
for some constants $q \in \Rb$ and ${\psi} \in \Rb^d$ and some functions $b, \beta_i  : [0,\Tb] \to  \mathbb{R}^d$  and $\ell,\lam_i : [0,\Tb] \to \mathbb{R}^{d \times d}$.  Note that $\sig^\text{Tr}$ denotes the transpose of $\sig$.

%
%

\section{Bond and option pricing}
\label{sec:review}
In this section we review some classical results {on} bond and option pricing in an affine short-rate setting.  Our aim here is not to be rigorous, but rather to present in a concise and \textit{formal} manner the results that will be needed in subsequent sections.  For a rigorous treatment of the formal results presented below, we refer the reader to \cite[Chapter 10]{filipovic2009term}.
\\[0.5em]
To begin, for any $T \leq \Tb$ and $\nu \in \Cb^d$, let us define $\Gam( \,\cdot\,,\,\cdot\,;T,\nu) : [0,T] \times \Rb^d \to \Cb^d$ by
\begin{align}
\Gam(t,Y_t;T,\nu) 
	:= \Eb_t \exp \Big( -\int_t^T r(Y_s) \dd s + \sum_{i=1}^d \nu_i Y_T^{(i)} \Big) , \label{eq:Gamma-def}
\end{align}
where we have introduced the short-hand notation $\Eb_t( \, \cdot \, ) := \Eb( \, \cdot \, | \Fc_t )$.  The existence of the function $\Gam$ follows from the Markov property of $Y$.  Formally, $\Gam$ satisfies the Kolmogorov backward partial differential equation (PDE)
\begin{align}
(\d_t + \Ac(t) - r ) \Gam(t,\, \cdot \, ; T,\nu)
	&=	0 , &
\Gam(T,y;T,\nu)
	&=	\exp \Big( \sum_{i=1}^d \nu_i y_i \Big) , \label{eq:Gamma-pde}
\end{align}
where the operator $\Ac$ is the generator of $Y$ under $\Pb$.  Explicitly, the generator $\Ac$ is given by
\begin{align}
\Ac(t)
	&=	\sum_{i=1}^d \mu_i(t,y) \d_{y_i} + \frac{1}{2} \sum_{i=1}^d \sum_{j=1}^d \Big( \sig(t,y) \sig^\text{Tr}(t,y) \Big)_{i,j} \d_{y_i}\d_{y_j} , \label{eq:A}
\end{align}
where $(\sig \sig^\text{Tr})_{i,j}$ denotes its $(i,j)$-th component of $\sig \sig^\text{Tr}$.  One can verify by direct substitution that the solution to \eqref{eq:Gamma-pde} is
\begin{align}
\Gam(t,y;T,\nu) 
	&=	\exp \Big( - F(t;T,\nu) - \sum_{i=1}^d G_i(t;T,\nu) y_i \Big) , \label{eq:Gamma-explicit}
\end{align}
where the functions $F$ and $G=(G_i)_{i=1,2,\ldots,d}$ are the solution of the following system of coupled ordinary differential equations (ODEs)
\begin{align}
\d_t F(t;T,\nu)
	&=	\frac{1}{2}G^{\text{Tr}}(t;T,\nu)\ell(t)G(t;T,\nu)-b^{\text{Tr}}(t)G(t;T,\nu)-q , &
F(T;T,\nu)
	&=	0 , \label{eq:F-ode} \\
\d_t G_i(t;T,\nu)
	&=	\frac{1}{2}G^{\text{Tr}}(t;T,\nu)\lam_i(t)G(t;T,\nu)-\beta_i^{\text{Tr}}(t)G(t;T,\nu)-{\psi}_i , &
G_i(T;T,\nu)
	&=	- \nu_i . \label{eq:G-ode}
\end{align}

\noindent
Now, for any $T \leq \Tb$, let us denote by $B^T = (B_t^T)_{0 \leq t \leq T}$ the value of a \textit{zero-coupon bond} that pays one unit of currency at time $T$.
In the absence of arbitrage the process $B^T/M$ must be a $(\Pb,\Fb)$-martingale.  As such, we have
\begin{align}
\frac{B_t^T}{M_t} 
	&= \Eb_t \Big( \frac{B_T^T}{M_T} \Big) = \Eb_t \Big( \frac{1}{M_T} \Big) ,
\end{align}
where we have used $B_T^T = 1$.  Solving for $B_t^T$, we obtain
\begin{align}
B_t^T
	&=	\Eb_t \Big( \frac{M_t}{M_T} \Big) = \Eb_t \Big( \ee^{- \int_t^T r(Y_s) \dd s} \Big) = \Gam(t,Y_t;T,0) \\
	&=	\exp \Big( - F(t;T,0) - \sum_{i=1}^d G_i(t;T,0) Y_t^{(i)} \Big) , \label{eq:B-explicit}
\end{align}
where the third equality follows from \eqref{eq:Gamma-def} and the fourth equality follows from \eqref{eq:Gamma-explicit}.
\\[0.5em]
Next, let $V = (V_t)_{0 \leq t \leq T}$ denote the value of a European option that pays $\phi( \log B_T^\Tb )$ at time $T$ for some function $\phi : \Rb_- \to \Rb$.  With the aim of finding $V_t$, let $\phih:\Cb \to \Cb$ denote the generalized Fourier transform of $\phi$, which {is defined} as follows
\begin{align}
\phih(\om)
	&:=	\int_{-\infty}^{\infty} \dd x \, \ee^{-\ii \om x} \phi(x) , &
\om
	&=	\om_r + \ii \om_i , &
\om_r, \om_i
	&\in \Rb . \label{eq:ft}
\end{align}
We can recover $\phi$ from $\phih$ using the inverse Fourier transform
\begin{align}
\phi(x)
	&:=	\frac{1}{2\pi} \int_{-\infty}^{\infty} \dd \om_r \, \ee^{\ii \om x} \phih(\om) . \label{eq:ift}
\end{align}
Noting that, in the absence of arbitrage, the process $V/M$ must be a $(\Pb,\Fb)$-martingale, we have
\begin{align}
\frac{V_t}{M_t} 
	&= \Eb_t \Big( \frac{V_T}{M_T} \Big) 
	= 	\Eb_t \Big( \frac{\phi(\log B_T^\Tb)}{M_T} \Big) .
\end{align}
Solving for $V_t$, we have that
\begin{align}
V_t
	&=	\Eb_t \exp \Big( -\int_t^T r(Y_s) \dd s \Big) \phi( \log B_T^\Tb ) \\
	&=	\frac{1}{2\pi} \int_{-\infty}^{\infty} \dd \om_r \, \phih(\om) \Eb_t \exp \Big( -\int_t^T r(Y_s) \dd s \Big) \exp( \ii \om \log B_T^\Tb ) \\
	&=	\frac{1}{2\pi} \int_{-\infty}^{\infty} \dd \om_r \, \phih(\om) \Eb_t \exp \Big( -\int_t^T r(Y_s) \dd s \Big) \Eb_T \exp( \ii \om \log B_T^\Tb ) \\
	&=	\frac{1}{2\pi} \int_{-\infty}^{\infty} \dd \om_r \, \phih(\om) \exp \Big( - \ii \om F(T;\Tb,0) \Big) 
			\Eb_t \exp \Big(  -\int_t^T r(Y_s) \dd s - \sum_{i=1}^d \ii \om G_i(T;\Tb,0) Y_T^{(i)} \Big) \\
	&=	\frac{1}{2\pi} \int_{-\infty}^{\infty} \dd \om_r \, \phih(\om) \exp \Big( - \ii \om F(T;\Tb,0) \Big) \Gam(t,Y_t;T,-\ii \om G(T;\Tb,0))
	=: 	u(t,Y_t;T,\Tb) , \label{eq:V-explicit}
\end{align}
where the second equality follows from \eqref{eq:ift}, the fourth follows from \eqref{eq:B-explicit} and the fifth follows from \eqref{eq:Gamma-def}.
For the particular case of a $T$-maturity European Call option written on $B^\Tb$ we have
\begin{align}
\phi(x)
	&=	( \ee^x - \ee^k )^+ , &
\phih(\om)
	&=	\frac{-\ee^{k- \ii k \om}}{\om^2 + \ii \om } , &
\om_i
	&<	-1 , \label{eq:call-ft}
\end{align}
where $k$ is the $\log$ of the strike.

%
%

\section{Relation to local-stochastic volatility models}
\label{sec:lsv}
While \eqref{eq:V-explicit} in conjunction with \eqref{eq:call-ft} can be used to compute $T$-maturity Call prices on $B^\Tb$, the resulting expression tells us very little about the corresponding implied volatilities.  In this section, we will establish a precise relation between affine short-rate models and local-stochastic volatility models.  This relation will be used in subsequent sections to find an explicit approximation for Call option implied volatilities.
\\[0.5em]
We begin deriving the dynamics of $B^T/M$.  Using \eqref{eq:dM} and \eqref{eq:B-explicit}, we have by It\^o's Lemma that 
\begin{align}
\dd \Big( \frac{B_t^T}{M_t} \Big)
	&=	\Big( \frac{B_t^T}{M_t} \Big) \sum_{j=1}^d \gam_j(t,Y_t;T) \dd W_t^{(j)} , \label{eq:BoverM}
\end{align}
where we have introduced
\begin{align}
\gam_j(t,Y_t;T)
	&:= 	\sum_{i=1}^d \sig_{i,j}(t,Y_t)  \d_{y_i} \log \Gamma(t,Y_t;T,0) 
	=		- \sum_{i=1}^d \sig_{i,j}(t,Y_t) G_i(t;T,0) . \label{eq:gamma-def}
\end{align}
Observe that $B^T/M$ is a $(\Pb,\Fb)$-martingale, as it must be.
\\[0.5em]
It will be helpful at this point to introduce the \textit{$T$-forward probability measure} $\Pbt$, whose relation to $\Pb$ is given by the following Radon-Nikodym derivative
\begin{align}
\frac{\dd \Pbt}{\dd \Pb} 
	&:=	\frac{M_0 B_T^T}{B_0^T M_T} 
	=	\exp \Big( - \frac{1}{2} \sum_{j=1}^d  \int_0^T \gam_j^2(t,Y_t;T) \dd t + \sum_{j=1}^d  \int_0^T \gam_j(t,Y_t;T) \dd W_t^{(j)} \Big) . \label{eq:measure-change}
\end{align}
Note that the the last equality follows from \eqref{eq:BoverM}.    The following lemma will be useful.

\begin{lemma}
Let $\Pi = (\Pi_t)_{0 \leq t \leq \Tb}$ denote the value of a self-financing portfolio and let $\Pi^T = (\Pi_t^T)_{0 \leq t \leq T}$, defined by $\Pi_t^T := \Pi_t/B_t^T$, be the \textit{$T$-forward price of $\Pi$}.  Then the process $\Pi^T$ is a $(\Pbt,\Fb)$-martingale.
\end{lemma}
\begin{proof}
Define the \textit{Radon-Nikodym derivative process} $Z = (Z_t)_{0 \leq t \leq T}$ by $Z_t := \Eb_t (\dd \Pbt/ \dd \Pb)$.  Using the fact that $\Pi/M$ is a $(\Pb,\Fb)$-martingale as well as \cite[Lemma 5.2.2]{shreve2004stochastic} we have for any $0 \leq t \leq s \leq T$ that
\begin{align}
\frac{\Pi_t}{M_t}
	&=	\Eb_t \Big( \frac{\Pi_s}{M_s} \Big)
	=		Z_t \Ebt_t \Big( \frac{1}{Z_s} \frac{\Pi_s}{M_s} \Big)
	=		\frac{B_t^T}{M_t} \Ebt_t \Big( \frac{M_s}{B_s^T} \frac{\Pi_s}{M_s} \Big) , \label{eq:1}
\end{align}
where $\Ebt$ denotes an expectation under $\Pbt$.  Dividing both sides of equation \eqref{eq:1} by $B_t^T$ and canceling common factors of $M_t$ and $M_s$, we obtain
\begin{align}
\Pi_t^T 
	&=	\frac{\Pi_t}{B_t^T}
	 =	\Ebt_t \frac{\Pi_s}{B_s^T}
	 =	\Ebt_t  \Pi_s^T  ,
\end{align}
which establishes that $\Pi^T$ is a $(\Pbt,\Fb)$-martingale, as claimed.
\end{proof}

\noindent
Now, let us denote by $X = (X_t)_{0 \leq t \leq T}$ the $\log$ of the $T$-forward price of a $\Tb$-maturity bond $B^\Tb$.  We have
\begin{align}
X_t
	&:=	\log \Big( \frac{B_t^\Tb}{B_t^T} \Big) \label{eq:X-def} \\
	&= 	F(t;T,0) - F(t;\Tb,0) + \sum_{i=1}^d \big( G_i(t;T,0) - G_i(t;\Tb,0) \big) Y_t^{(i)}  , \label{eq:X=Y}
\end{align}
where the second equality follows from \eqref{eq:B-explicit}.  It follows from the explicit relationship \eqref{eq:X=Y} between $X$ and $Y$ that the process $(X,\Yt) := (X_t,Y_t^{(2)},\ldots,Y_t^{(d)})_{0 \leq t \leq T}$ is a $d$-dimensional Markov process.  We are now in a position to state the main result of this section.

\begin{proposition}
Let $V^T = V/B^T$ denote the $T$-forward price of an option that pays $\phi( \log B_T^\Tb)$ at time $T$.
Then there exists a function $v(\,\cdot\,,\,\cdot\,,\,\cdot\,;T,\Tb): [0,T] \times \Rb_- \times \Rb^{d-1} \to \Rb$ such that 
\begin{align}
V_t^T
	&=	v(t,X_t,\Yt_t;T,\Tb) .
\end{align}
Moreover, 
the function $v$ satisfies the following PDE
\begin{align}
(\d_t + \Act(t)) v(t,\,\cdot\,,\,\cdot\,;T,\Tb)
	&=	0 , &
{v(T,x,\yt;T,\Tb)}
	&=	\phi(x) , \label{eq:v-pde}
\end{align}
where $\Act$ is the generator of $(X,\Yt)$ under $\Pbt$.  Explicitly, $\Act$ is given by
\begin{align}
\Act(t)
	&= \frac{1}{2} \sum_{i=1}^d \sum_{j=1}^d \Big( \sigt(t,x,\yt;T,\Tb) \sigt^\text{Tr}(t,x,\yt;T,\Tb) \Big)_{i,j} \\ &\quad
			\times \Big( G_i(t;T,0) - G_i(t;\Tb,0) \Big) \Big( G_j(t;T,0) - G_j(t;\Tb,0) \Big) (\d_x^2-\d_x)  \\ & \quad
			+\sum_{i=2}^d \Big( \mut_i(t,x,\yt;T,\Tb) - \sum_{j=1}^d \Big( \sigt(t,x,\yt;T,\Tb) \sigt^\text{Tr}(t,x,\yt;T,\Tb) \Big)_{i,j} G_j(t;T,0) \Big) \d_{y_i} \\&\quad
			+ \frac{1}{2} \sum_{i=2}^d \sum_{j=2}^d \Big( \sigt(t,x,\yt;T,\Tb) \sigt^\text{Tr}(t,x,\yt;T,\Tb) \Big)_{i,j} \d_{y_i}\d_{y_j}  \\ & \quad
			+ \sum_{i=2}^d \sum_{j=1}^d \Big( \sigt(t,x,\yt;T,\Tb) \sigt^\text{Tr}(t,x,\yt;T,\Tb) \Big)_{i,j} \Big( G_j(t;T,0) - G_j(t;\Tb,0) \Big) \d_x \d_{y_i} ,
		\label{eq:A-tilde}
\end{align}
where the functions
$\mut(\,\cdot\,,\,\cdot\,,\,\cdot\,;T,\Tb) : [0,T] \times \Rb_- \times \Rb^{d-1} \to \Rb^d$ and 
$\sigt(\,\cdot\,,\,\cdot\,,\,\cdot\,;T,\Tb) : [0,T] \times \Rb_- \times \Rb^{d-1} \to \Rb^{d \times d}$
are given by
\begin{align}
\mut(t,x,\yt;T,\Tb)
	&:= \mu(t,\eta(t,x,\yt;T,\Tb),\yt) , &
\sigt(t,x,\yt;T,\Tb)
	&:= \sig(t,\eta(t,x,\yt;T,\Tb),\yt) , \label{eq:mu-tilde-sigma-tilde}
\end{align}
the function $\eta(\,\cdot\,,\,\cdot\,,\,\cdot\,;T,\Tb) : [0,T] \times \Rb_- \times \Rb^{d-1} \to \Rb$ is defined as follows
\begin{align}
\eta(t,x,\yt;T,\Tb)
	&=	\frac{F(t;T,0) - F(t;\Tb,0) - x  + \sum_{i=2}^d (G_i(t;T,0) - G_i(t;\Tb,0) ) y_i}{G_1(t;\Tb,0)-G_1(t;T,0)} ,\label{eq:eta-def}
\end{align}
and the functions $F$ and $G_i$ satisfy the system of coupled ODEs \eqref{eq:F-ode} and \eqref{eq:G-ode}. 
\end{proposition}

\begin{proof}
Noting that $V^T$ is a $(\Pbt,\Fb)$-martingale, we have
\begin{align}
V_t^T
	&=	\frac{V_t}{B_t^T}
	 =	\Ebt_t \big(  \frac{V_T}{B_T^T} \Big)
	 =	\Ebt_t\phi(\log B_T^\Tb)
	 =	\Ebt_t\phi(X_T)
	 =:	v(t,X_t{,\Yt_t};T,\Tb) ,
\end{align}
where the existence of the function $v$ follows from the Markov property of $(X,\Yt)$.  The function $v$ satisfies the {Kolmogorov} backward PDE \eqref{eq:v-pde} where $\Act$ denotes the generator of $(X,\Yt)$ under $\Pbt$.  To derive the expression \eqref{eq:A-tilde} for $\Act$, we note that, by Girsanov's theorem and \eqref{eq:measure-change}, the process $\Wt := (\Wt_t^{(1)}, \Wt_t^{(2)}, \ldots, \Wt_t^{(d)})_{0 \leq t \leq T}$, defined as follows
\begin{align}
\Wt_t^{(j)}
	&:=	- \int_0^t \gam_j(s,Y_s;T) \dd s + W_t^{(j)} , \label{eq:W_tilde}
\end{align}
is a $d$-dimensional $(\Pbt,\Fb)$-Brownian motion.  Thus, we have from equations \eqref{eq:dYi_t}, \eqref{eq:gamma-def} and \eqref{eq:W_tilde} that
\begin{align}
\dd Y_t^{(i)}
	&=	\Big( \mu_i(t,Y_t) - \sum_{j=1}^d \Big( \sig(t,Y_t) \sig^\text{Tr}(t,Y_t) \Big)_{i,j} G_j(t;T,0)  \Big) \dd t 
			+ \sum_{j=1}^d \sig_{i,j}(t,Y_t)\dd \Wt_t^{(j)} \\
	&=	\Big( \mut_i(t,X_t,{\Yt_t};T,\Tb) - \sum_{j=1}^d \Big( {\sigt(t,X_t,\Yt_t;T,\Tb)} \sigt^\text{Tr}(t,X_t,{\Yt_t};T,\Tb) \Big)_{i,j} G_j(t;T,0)  \Big) \dd t \\ & \quad
			+ \sum_{j=1}^d \sigt_{i,j}(t,X_t,{\Yt_t};T,\Tb)\dd \Wt_t^{(j)} , \label{eq:system_Ahat}
\end{align}
where, in the the second equality, we have used $Y_t^{(1)} = \eta(t,X_t,\Yt_t;T,\Tb)$, which follows from from \eqref{eq:X=Y}.
Similarly, using \eqref{eq:BoverM} and \eqref{eq:X-def}, we find using It\^o's Lemma that
\begin{align}
\dd X_t 
	&= -\frac{1}{2} \sum_{i=1}^d \sum_{j=1}^d \Big( \sig(t,Y_t) \sig^\text{Tr}(t,Y_t) \Big)_{i,j} \\ & \quad
			\times \Big( G_i(t;T,0) - G_i(t;\Tb,0) \Big) \Big( G_j(t;T,0) - G_j(t;\Tb,0) \Big) \dd t \\ & \quad
		+ \sum_{i=1}^d \sum_{j=1}^d  \sig_{i,j}(t,Y_t) \Big( G_i(t;T,0) - G_i(t;\Tb,0) \Big) \dd \Wt_t^{(j)} \\
	&= -\frac{1}{2} \sum_{i=1}^d \sum_{j=1}^d \Big( \sigt(t,X_t,{\Yt_t};T,\Tb) \sigt^\text{Tr}(t,X_t,{\Yt_t};T,\Tb) \Big)_{i,j} \\ & \quad
			\times \Big( G_i(t;T,0) - G_i(t;\Tb,0) \Big) \Big( G_j(t;T,0) - G_j(t;\Tb,0) \Big) \dd t \\ & \quad
		+ \sum_{i=1}^d \sum_{j=1}^d  \sigt_{i,j}(t,X_t,{\Yt_t};T,\Tb) \Big( G_i(t;T,0) - G_i(t;\Tb,0) \Big) \dd \Wt_t^{(j)} . \label{eq:dX}
\end{align}
The explicit expression \eqref{eq:A-tilde} for the generator $\Act$ follows from \eqref{eq:system_Ahat} and \eqref{eq:dX}.
\end{proof}

\noindent
Observe that $\ee^X = B^\Tb/B^T$ is a strictly positive $(\Pbt,\Fb)$-martingale.  Thus, the process $(X,\Yt)$ has the same form as a local-stochastic volatility model where $X$ represents the $\log$ of the $T$-forward price of an risky asset (e.g., stock, index, etc.) and $\Yt$ represents $(d-1)$ non-local factors of volatility.

\begin{example}
Consider a one-factor affine {short-rate} model ($d = 1$).  Then $X$ has the form of a (pure) local volatility model with generator
\begin{align}
\Act(t)	
	&=	c(t,x) (\d_x^2 - \d_x ) , &
c(t,x)
	&:=	\tfrac{1}{2} \sigt^2(t,x;T,\Tb) {\Big(  G(t;T,0) - G(t;\Tb,0) \Big)}^2, \label{eq:A-1d}
\end{align}
where we have omitted the argument $\yt$ as it plays no role.
\end{example}

\begin{example}
Consider a two-factor affine short-rate model ($d=2$). Then the process $(X,Y^{(2)})$ has the form of a local-stochastic volatility model with a single non-local factor of volatility.  The generator in this case, is given by
\begin{align}
\Act(t)	
	&=	c(t,x,y_2) (\d_x^2 - \d_x ) + f(t,x,y_2) \d_{y_2} + g(t,x,y_2) \d_{y_2}^2  + h(t,x,y_2) \d_x \d_{y_2}, \label{eq:A-2d}
\end{align}
where the functions $c$, $f$, $g$ and $h$ are given by 
\begin{align}
c(t,x,y_2)
	&:=	 \tfrac{1}{2}{\Big(}\sigt^2_{1,1}(t,x,y_2;T,\Tb)+\sigt^2_{1,2}(t,x,y_2;T,\Tb){\Big)} {\Big(}G_1(t;T,0) - G_1(t;\Tb,0){\Big)}^2  \\ & \quad
			+	{\Big(}\sigt_{1,1}(t,x,y_2;T,\Tb)\sigt_{2,1}(t,x,y_2;T,\Tb)+\sigt_{1,2}(t,x,y_2;T,\Tb)\sigt_{2,2}(t,x,y_2;T,\Tb){\Big)} \\ & \quad
				\times {\Big(}G_1(t;T,0) - G_1(t;\Tb,0){\Big)}{\Big(}G_2(t;T,0) - G_2(t;\Tb,0){\Big)} \\ & \quad
			+ {\tfrac{1}{2}}{\Big(}\sigt^2_{2,1}(t,x,y_2;T,\Tb)+\sigt^2_{2,2}(t,x,y_2;T,\Tb){\Big)} {\Big(}G_2(t;T,0) - G_2(t;\Tb,0){\Big)}^2, \\
f(t,x,y_2)
	&:=	\mut_2(t,x,y_2;T,\Tb)  -{\Big(}\sigt^2_{2,1}(t,x,y_2;T,\Tb)+\sigt^2_{2,2}(t,x,y_2;T,\Tb){\Big)}G_2(t;T,0) \\ & \quad
			- {\Big(}\sigt_{1,1}(t,x,y_2;T,\Tb)\sigt_{2,1}(t,x,y_2;T,\Tb) \\ & \quad
			+ \sigt_{1,2}(t,x,y_2;T,\Tb)\sigt_{2,2}(t,x,y_2;T,\Tb){\Big)}G_1(t;T,0), \\
g(t,x,y_2)
	&:=	\tfrac{1}{2}{\Big(}\sigt^2_{2,1}(t,x,y_2;T,\Tb) + \sigt^2_{2,2}(t,x,y_2;T,\Tb){\Big)}, \\
h(t,x,y_2)
	&:= {\Big(}\sigt^2_{2,1}(t,x,y_2;T,\Tb)+\sigt^2_{2,2}(t,x,y_2;T,\Tb){\Big)} {\Big(}G_2(t;T,0) - G_2(t;\Tb,0){\Big)} \\ & \quad
			+ {\Big(}\sigt_{1,1}(t,x,y_2;T,\Tb)\sigt_{2,1}(t,x,y_2;T,\Tb) \\ & \quad
			+  \sigt_{1,2}(t,x,y_2;T,\Tb)\sigt_{2,2}(t,x,y_2;T,\Tb){\Big)} {\Big(}G_1(t;T,0) - G_1(t;\Tb,0){\Big)}.
\end{align}
\end{example}


%
%

\section{Option price asymptotics}
\label{sec:price-asymptotics}
We have from \eqref{eq:v-pde} that $v$ satisfies a parabolic PDE of the form
\begin{align}
( \d_t + \Act(t) ) v(t, \,\cdot \,)
	&=	0 , &
\Act(t)
	&=	\sum_{|\alpha| \leq 2} a_\alpha(t,z) \d_z^\alpha , &
v(T, \, \cdot \,)
	&=	\phi , \label{eq:pde-form}
\end{align}
where $z := (x, y_2, \ldots, y_d)$.  Note that, for brevity, we have omitted the dependence on $T$ and $\Tb$ and we have introduced standard multi-index notation
\begin{align}
\alpha
	&=	(\alpha_1, \alpha_2, \dots, \alpha_d) , &
\d_z^\alpha
	&=	\prod_{i=1}^d \d_{z_i}^{\alpha_i} , &
z^\alpha
	&=	\prod_{i=1}^d {z_i}^{\alpha_i} , &
| \alpha |
	&=	\sum_{i=1}^d \alpha_i , &
\alpha!
	&=	\prod_{i=1}^d \alpha_i! .
\end{align}
In general there is no explicit solution to PDEs of the form \eqref{eq:pde-form}.  In this section, we will show in a formal manner how an explicit approximation of $v$ can be obtained by using a simple Taylor series expansion of the coefficients $a_\alpha$ of $\Act$.  The method described below was introduced for scalar diffusions in \cite{pagliarani2011analytical} and subsequently extended to $d$-dimensional diffusions in \cite{lorig-pagliarani-pascucci-2,lorig-pagliarani-pascucci-4}.  
\\[0.5em]
To begin, for any $\eps \in [0,1]$ and $\zb: [0,T] \to \Rb^d$, let $v^\eps$ be the unique classical solution to
\begin{align}
0
	&=	( \d_t + \Act^\eps(t) ) v^\eps(t, \,\cdot \,) , &
v^\eps(T, \, \cdot \,)
	&=	\phi , \label{eq:v-eps-pde} 
\end{align}
where the operator $\Act^\eps$ is defined as follows
\begin{align}
\Act^\eps(t)
	&:=	\sum_{|\alpha| \leq 2}  a_\alpha^\eps(t,z) \d_z^\alpha , &
	&\text{with}&
a_\alpha^\eps
	&:=	a_\alpha(t,\zb(t) + \eps(z - \zb(t))) , \label{eq:A-eps}
\end{align}
Observe that $\Act^\eps |_{\eps = 1} = \Act$ and thus $v^\eps |_{\eps=1} = v$.  We will seek an approximate solution of \eqref{eq:v-eps-pde} by expanding $v^\eps$ and $\Act^\eps$ in powers of $\eps$.  Our approximation for $v$ will be obtained by setting $\eps = 1$ in our approximation for $v^\eps$.   We have
\begin{align}
v^\eps
	&=	\sum_{n=0}^\infty \eps^n v_n , &
\Act^\eps(t)
	&=	\sum_{n=0}^\infty \eps^n \Act_n(t) , \label{eq:expansion}
\end{align}
where the functions $(v_n)$ are, at the moment, unknown, and the operators $(\Act_n)$ are given by
\begin{align}
\Act_n(t)
	&=	\frac{\dd^n }{\dd \eps^n} \Act^\eps |_{\eps=0}
	=		\sum_{|\alpha| \leq 2} a_{\alpha,n}(t,z) \d_z^\alpha , &
a_{\alpha,n}
	=		\sum_{|\beta|=n} \frac{1}{\beta!} (z - \zb(t))^\beta \d_z^\beta a_\alpha(t,\zb(t))  .
\end{align}
Note that $a_{\alpha,n}(t, \, \cdot \,)$ is the sum of the $n$th order terms in the Taylor series expansion of $a_\alpha(t,\,\cdot\,)$ about the point $\zb(t)$.  Inserting the expansions from \eqref{eq:expansion} for $v^\eps$ and $\Act^\eps$ into PDE \eqref{eq:v-eps-pde} and collecting terms of like order in $\eps$ we obtain
\begin{align}
&\Oc(\eps^0):&
0
	&=	( \d_t + \Act_0(t) ) v_0(t, \,\cdot \,) , &
v_0(T, \, \cdot \,)
	&=	\phi  , \label{eq:v0-pde} \\
&\Oc(\eps^n):&
0
	&=	( \d_t + \Act_0(t) ) v_n(t, \,\cdot \,)  + \sum_{k=1}^n \Act_k(t) v_{n-k}(t,\,\cdot\,) , &
v_n(T, \, \cdot \,)
	&=	0  . \label{eq:vn-pde}
\end{align}
Now, observe that the coefficients $(a_{\alpha,0})$ of $\Act_0$ do not depend on $z$.  Thus, $\Act_0$ is the generator of a $d$-dimensional Brownian motion with a time-dependent drift vector and covariance matrix.  As such, $v_0$ is given by
\begin{align}
v_0(t,z)
	&=	\Pc_0(t,T)\varphi(z)
	=		\int_{\Rb^d} \dd z' \, p_0(t,z;T,z') \phi(z') . \label{eq:v0-explicit}
\end{align}
where $\Pc_0$ is the semigroup generated by $\Act_0$ and $p_0$ is the associated transition density (i.e., the solution to \eqref{eq:v0-pde} with $\phi = \del_{z'}$).  Explicitly, we have
\begin{align}
p_0(t,z;T,z')
	&=	\tfrac{1}{\sqrt{(2\pi)^d|\Cv(t,T)|}}
			{\exp\left(-\frac{1}{2} (z'-z-\mv(t,T))^\text{Tr} \Cv^{-1}(t,T) (z'-z-\mv(t,T)) \right)} , \label{eq:p0}
\end{align}
where $\mv$ and $\Cv$ are given by
\begin{align}
\mv(t,T)
		&:=	\int_t^T \dd s \, m(s) , &
\Cv(t,T)
		&:=	\int_t^T \dd s \, A(s) , \label{eq:m-and-C}
\end{align}
and $m$ and $A$ are, respectively, the instantaneous drift vector and covariance matrix
\begin{align}
m(s)
		&:=	\begin{pmatrix}
				a_{(1,0,\cdots,0),0}(s) \\ a_{(0,1,\cdots,0),0}(s) \\ \vdots \\  a_{(0,0,\cdots,1),0}(s)
				\end{pmatrix} , &
A(s)
		&:= \begin{pmatrix}
				2a_{(2,0,\cdots,0),0}(s) & a_{(1,1,\cdots,0),0}(s) & \ldots &  {a_{(1,0,\cdots,1),0}(s)} \\
				a_{(1,1,\cdots,0),0}(s) & 2a_{(0,2,\cdots,0),0}(s) & \ldots &  a_{(0,1,\cdots,1),0}(s) \\
				\vdots & \vdots & \ddots & \vdots \\
				a_{(1,0,\cdots,1),0}(s) & a_{(0,1,\cdots,1),0}(s) & \ldots &  2 a_{(0,0,\cdots,2),0}(s) \\
				\end{pmatrix} .
\end{align}
By Duhamel's principle, the solution $v_n$ of \eqref{eq:vn-pde} is
\begin{align}
v_n(t,z)
	&=	\sum_{k=1}^n \int_t^T \dd t_1 \, \Pc_0(t,t_1) \Act_k(t_1) v_{n-k}(t_1,z) \\
	&=	\sum_{k=1}^n \sum_{i \in I_{n,k}}
      \int_{t}^T \dd t_1 \int_{t_1}^T \dd t_2 \cdots \int_{t_{k-1}}^T \dd t_k \\ & \qquad 
       \Pc_0(t,t_1) \Ac_{i_1}(t_1)
       \Pc_0(t_1,t_2) \Ac_{i_2}(t_2) \cdots
       \Pc_0(t_{k-1},t_k) \Ac_{i_k}(t_k)
       \Pc_0(t_k,T)\phi(z) , \label{eq:vn-explicit} \\
I_{n,k}
    &= \{ i = (i_1, i_2, \cdots , i_k ) \in \mathds{N}^k : i_1 + i_2 + \cdots + i_k = n \} .
            \label{eq:Ink}
\end{align}
While the expression \eqref{eq:vn-explicit} for $v_n$ is explicit, it is not easy to compute as operating on a function with $\Pc_0$ requires performing a $d$-dimensional integral.  The following proposition establishes that $v_n$ can be expressed as a differential operator acting on $v_0$.

\begin{proposition}
\label{thm:vn}
The solution $v_n$ of PDE \eqref{eq:vn-pde} is given by
\begin{align}
v_n(t,z)
    &=  \Lc_n(t,T) v_0(t,z) , \label{eq:un} 
\end{align}
where $\Lc$ is a linear differential operator, which is given by
\begin{align}
\Lc_n(t,T)
    &=  \sum_{k=1}^n \sum_{i \in I_{n,k}}
        \int_{t}^T \dd t_1 \int_{t_1}^T \dd t_2 \cdots \int_{t_{k-1}}^T \dd t_k
        \Gc_{i_1}(t,t_1)
        \Gc_{i_2}(t,t_2) \cdots
         \Gc_{i_k}(t,t_k) , \label{eq:Ln}
\end{align}
the index set $I_{n,k}$ as defined in \eqref{eq:Ink} and the operator $\Gc_i$ is given by
\begin{align}
\Gc_i(t,t_k) 
	&:= 	\sum_{|\alpha |\leq 2}  a_{\alpha,i}(t_k,\Zc(t,t_k)) \d_z^\alpha , &
\Zc(t,t_k)
  &:= z + \mv(t,t_k) + \Cv(t,t_k) \nabla_z . 	\label{eq:Gc.def}
\end{align}
\end{proposition}

\begin{proof}
The proof, which is given in \cite[Theorem 2.6]{lorig-pagliarani-pascucci-2}, relies on the fact that, for any $0 \leq t \leq t_k < \infty$ the operator $\Gc_i$ in \eqref{eq:Gc.def} satisfies
\begin{align}
\Pc_0(t,t_k) \Ac_{i}(t_k)
    &=  \Gc_i(t,t_k) \Pc_0(t,t_k) . \label{eq:PA=GP}
\end{align}
Using \eqref{eq:PA=GP}, as well as the semigroup property ${\Pc_0}(t_1,t_2) {\Pc_0}(t_2,t_3) = {\Pc_0}(t_1,{t_3})$, we have that
\begin{align}
&\Pc_0(t,t_1) \Ac_{i_1}(t_1) \Pc_0(t_1,t_2) \Ac_{i_2}(t_2) \cdots \Pc_0(t_{k-1},t_k) \Ac_{i_k}(t_k) \Pc_0(t_k,T) \phi(z) \\
	&=	\Gc_{i_1}(t,t_1) \Gc_{i_2}(t,t_2) \cdots \Gc_{i_k}(t,t_k)
       \Pc_0(t,t_1) \Pc_0(t_1,t_2) \cdots \Pc_0(t_{k-1},t_k) \Pc_0(t_k,T) \phi \\
	&=	\Gc_{i_1}(t,t_1) \Gc_{i_2}(t,t_2) \cdots \Gc_{i_k}(t,t_k) \Pc_0(t,T) \phi \\
	&=	\Gc_{i_1}(t,t_1) \Gc_{i_2}(t,t_2) \cdots \Gc_{i_k}(t,t_k) v_0(t,\,\cdot\,) , \label{eq:result}
\end{align}
where, in the last equality we have used $\Pc_0(t,T) \varphi = v_0(t,\,\cdot\,)$.  Inserting \eqref{eq:result} into \eqref{eq:vn-explicit} yields \eqref{eq:un}.
\end{proof}

\noindent
Having obtained explicit expressions for the functions $(v_n)$, we define $\vb$, the \textit{$n$th order approximation of $v$}, as follows
\begin{align}
\vb_n
	:= \sum_{k=0}^n {v_k} . \label{eq:def-vbar}
\end{align}
Note that $\vb_n$ depends on the choice of $\zb$.  In general, if one is interested in the value of $v(t,z)$ a good choice for $\zb$ is $\zb(t) = z$.  Indeed, when one chooses $\zb(t) = z$, we have from \cite[Theorem 3.10]{lorig-pagliarani-pascucci-4} that
\begin{align}
|v(t,z) - \vb_n(t,z)|
	&=	\Oc( (T-t)^{(n+k+2)/2} ) &
	&\text{as $T-t \to 0$} , \label{eq:accuracy}
\end{align}
when the terminal data $\phi$ is a bounded function with globally Lipschitz continuous derivatives of order less than or equal to $k$.

%
%

\section{Implied volatility asymptotics}
\label{sec:imp-vol}
The goal of this section is to find an explicit approximation for the implied volatility corresponding to the $T$-forward Call price $v(t,x,\yt;T,\Tb,k)$ where we have included now the dependence on the $\log$ strike $k$.  For brevity, in what follows, we will omit the dependence on   $(t,x,\yt;T,\Tb,k)$.  
\\[0.5em]
To begin, we remind the reader that, in the Black-Scholes setting, the $T$-forward price of a risky asset $S$ has dynamics of the form
\begin{align}
\dd \Big( \frac{S_t}{B_t^T} \Big)
	&=	\Sig \Big( \frac{S_t}{B_t^T} \Big) \dd \Wt_t ,
\end{align}
where $\Sig > 0$ is the Black-Scholes volatility and $\Wt$ is a one-dimensional Brownian motion under $\Pbt$.
Given that $\log (S_t/B_T^T) = x$, the $T$-forward \textit{Black-Scholes {Call} price} with volatility $\Sig > 0$ is given by
\begin{align}
v^\BS(\Sig)
	&:=	\ee^x \Phi(d_+) - \ee^k \Phi(d_-) , &
d_\pm
   &= \frac{1}{\Sig \sqrt{T-t}} \left( x - k \pm \frac{\Sig^2 (T-t)}{2}  \right) , &
\Phi(d)
	&=	\int_{-\infty}^d \dd x \, \frac{1}{\sqrt{2\pi}} \ee^{-x^2/2} .
\end{align}
From this, one defines the \textit{implied volatility} corresponding to the $T$-forward {Call} price $v$ as the unique positive solution $\Sig$ to
\begin{align}
v^\BS(\Sig)
	&=	v . \label{eq:IV-def}
\end{align}
As in the previous section, we will seek an approximation of the implied volatility $\Sig^\eps$ corresponding to $v^\eps$ by expanding $\Sig^\eps$ in power of $\eps$.  Our approximation {of} $\Sig$ will then be obtained by setting $\eps = 1$.  We have
\begin{align}
\Sig^\eps
	&=	\Sig_0 + \del \Sig^\eps , &
\del \Sig^\eps
	&=	\sum_{n=1}^\infty \eps^n \Sig_n . \label{eq:I-expand}
\end{align}
Next, expanding $v^\BS(\Sig^\eps)$ in powers of $\eps$ we obtain
\begin{align}
v^\BS(\Sig^\eps)
	&=	v^\BS(\Sig_0 + \del \Sig^\eps) \\
	&=		\sum_{k=0}^\infty \frac{1}{k!}(\del \Sig^\eps \d_\Sig )^k v^\BS(\Sig_0) \\
	&=		v^\BS(\Sig_0) + 
						\sum_{k=1}^\infty \frac{1}{k!}  
						\sum_{n=1}^\infty \eps^n \sum_{I_{n,k}} \Big( \prod_{j=1}^k \Sig_{i_j} \Big) \d_\Sig^k v^\BS(\Sig_0) \\
	&=		v^\BS(\Sig_0) + 
						\sum_{n=1}^\infty \eps^n  \sum_{k=1}^\infty \frac{1}{k!}  
						\sum_{ I_{n,k}} \Big( \prod_{j=1}^k \Sig_{i_j} \Big) \d_\Sig^k v^\BS(\Sig_0) \\
	&=		v^\BS(\Sig_0) + 
						\sum_{n=1}^\infty \eps^n \bigg( \Sig_n \d_\Sig + \sum_{k=2}^\infty \frac{1}{k!}  
						\sum_{ I_{n,k}} \Big( \prod_{j=1}^k \Sig_{i_j} \Big) \d_\Sig^k  \bigg) v^\BS(\Sig_0) ,
\end{align}
where $I_{n,k}$ is given by \eqref{eq:Ink}.  Inserting the expansions for $v^\eps$ and $v^\BS(\Sig^\eps)$  into $v^\eps = v^\BS(\Sig^\eps)$ and collecting terms of like order in $\eps$ we obtain
\begin{align}
&\Oc(\eps^0)&
v_0
	&=	v^\BS(\Sig_0) , \label{eq:v0=expression} \\
&\Oc(\eps^n)&
v_n
	&=	\bigg( \Sig_n \d_\Sig + \sum_{k=2}^\infty \frac{1}{k!}  \sum_{ I_{n,k}} \Big( \prod_{j=1}^k \Sig_{i_j} \Big) \d_\Sig^k  \bigg) v^\BS(\Sig_0) . \label{eq:vn=expression}
\end{align}
Now, from \eqref{eq:v0-explicit} we have
\begin{align}
v_0
	&=	v^\BS \left( \sqrt{ \Cv_{1,1}(t,T)/(T-t) } \right) ,
\end{align}
where $\Cv$ is defined in \eqref{eq:m-and-C}.  Thus, it follows from \eqref{eq:v0=expression} that
\begin{align}
\Sig_0
	&=	\sqrt{ \Cv_{1,1}(t,T)/(T-t) } .\label{eq:sig-0}
\end{align}
Having identified $\Sig_0$, we can use \eqref{eq:vn=expression} to obtain $\Sig_n$ recursively for every $n \geq 1$.  We have
\begin{align}
\Sig_n
	&=	\frac{1}{\d_\Sig v^\BS(\Sig_0)} \bigg( v_n - \sum_{k=2}^\infty \frac{1}{k!}  \sum_{ I_{n,k}} \Big( \prod_{j=1}^k \Sig_{i_j} \Big) \d_\Sig^k v^\BS(\Sig_0) \bigg) . \label{eq:sig-n}
\end{align}
Using the expression given in \eqref{eq:un} for $v_n$, one can show that $\Sig_n$ is an $n$th order polynomial in $\log$-moneyness $k-x$ with coefficients that depend on $(t,T)$; see \cite[Section 3]{lorig-pagliarani-pascucci-2} for details.  We provide explicit expressions for $\Sig_0$, $\Sig_1$, and $\Sig_2$ for the cases $d=\{1,2\}$ in Appendix \ref{sec:explicit-expressions}.
\\[0.5em]
Having obtained expressions for $(\Sig_n)$, we define $\Sigb$, the \textit{$n$th order approximation of $\Sig$}, as follows
\begin{align}
\Sigb_n
	&:= \sum_{k=0}^n \Sig_k . \label{eq:def-sigbar}
\end{align}
Note that $\Sigb_n$ depends on the choice of $\zb$.  In general, the best choice for $\zb$ is $\zb(t) = (x,\yt)$. In this case, we have under mild conditions on the generator $\Act$ that
\begin{align}
|\Sig(t,x,\yt;T,\Tb,k)-\Sigb_n(t,x,\yt;T,\Tb,k)| 
	&= \Oc((T-t)^{(n+1)/2}), &
		&\text{as}&
|k-x| &= \Oc(  \sqrt{T-t}).
\end{align}
by \cite[Theorem 5.1]{pagliarani2017exact}.

%
%

\section{Examples}
\label{sec:examples}
In this section we use the results from Section \ref{sec:imp-vol} to compute approximate implied volatilities for $T$-forward Call prices written on $B^\Tb$ for the following four affine short-rate models:
\begin{itemize*}
\item Section \ref{sec:vasicek}: Vasicek model,
\item Section \ref{sec:cir}: Cox-{Ingersoll}-Ross model,
\item Section \ref{sec:cir2}: Two-factor Cox-{Ingersoll}-Ross model,
\item Section \ref{sec:fong}: Fong-Vasicek model.
\end{itemize*}
Note that, given $(X_t,\Yt_t) = (x,\yt)$, exact $T$-forward Call prices can be computed using
\begin{align}
v(t,x,\yt;T,\Tb)
	&=	\frac{u(t,y;T,\Tb)}{\Gam(t,y;T,0)} , &
y_1
	&=	\eta(t,x,\yt;T,\Tb) , \label{eq:v-exact}
\end{align}
where $\Gam$, $u$ and $\eta$ are given in \eqref{eq:Gamma-explicit}, \eqref{eq:V-explicit}-\eqref{eq:call-ft} and \eqref{eq:eta-def}, respectively.  The corresponding ``exact'' implied volatilities can be obtained by inserting \eqref{eq:v-exact} into \eqref{eq:IV-def} and solving for $\Sig$ numerically.  We will use this in what follows below in order to gauge the numerical accuracy of our implied volatility approximation $\Sigb_n$.

\subsection{Vasicek}
\label{sec:vasicek}
In the short-rate model developed in \cite{vasicek1977equilibrium}, the dynamics of $R=r(Y)$ are given by
\begin{align}
\dd Y_t
	&=	\kappa ( \theta - Y_t) \dd t + \del \dd W_t , &
R_t
	&=	Y_t . \label{eq:vasicek-model}
\end{align}
Comparing \eqref{eq:vasicek-model} with \eqref{eq:R_t} and \eqref{eq:dYi_t}, we see that the functions $r$, $\mu$, and $\sig$ are given by
\begin{align}
r(y)
	&=  y , &
\mu(t,y)
	&=	\kappa ( \theta - y) , &
\sig(t,y)
	&=	\del , \label{eq:r-mu-sigma-vasicek}
\end{align}
and comparing \eqref{eq:r-mu-sigma-vasicek} with \eqref{eq:r-mu-sigma} we identify
\begin{align}
q
	&=	0 , &
\psi
	&=	1 , &
b(t)
	&=	\kappa \theta , &
\beta(t)
	&=	- \kappa , &
\ell(t)
	&=	\del^2, &
\lam(t)
	&=	0 , 
\end{align}
where we have dropped the subscripts from $\psi$, $\beta$ and $\lam$ as $d=1$.
With the above parameters, the solution $G$ of ODE \eqref{eq:G-ode} is
\begin{align}
G(t;T,\nu)
	&=	- \ee^{-\kappa (T-t)}\nu + \frac{1-\ee^{-\kappa (T-t)}}{\kappa} . \label{eq:G-vasicek}
\end{align}
While the solution $F$ of ODE \eqref{eq:F-ode} is needed to compute exact Call option prices, we shall see that it is not needed to compute implied volatilities in the Vasicek setting.  As such, we do not provide a formula for $F$ here.  From \eqref{eq:mu-tilde-sigma-tilde}, \eqref{eq:eta-def}, and \eqref{eq:r-mu-sigma-vasicek}, we have
\begin{align}
\sigt(t,x;T,\Tb)
	&:= \del . \label{eq:sigma-tilde-vasicek}
\end{align}
And thus, using \eqref{eq:A-1d}, \eqref{eq:G-vasicek} and \eqref{eq:sigma-tilde-vasicek}, the generator $\Act$ is given by
\begin{align}
\Act(t)	
	&=	c(t,x) (\d_x^2 - \d_x ) , &
c(t,x)
	&=	\frac{1}{2} \del^2 \Big( \frac{1-\ee^{-\kappa (T-t)}}{\kappa} - \frac{1-\ee^{-\kappa (\Tb - t)}}{\kappa} \Big)^2 .
\end{align}
The explicit implied volatility approximation $\Sigb_n$ up to order $n=2$ can now be computed using the formulas in Appendix \ref{sec:explicit-expressions}.
Because the coefficient $c$ does not depend on $x$ in the Vasicek setting, the zeroth order implied volatility approximation is exact
\begin{align}
\Sig = \Sig_0
	&=		\sqrt{ \frac{1}{T-t} \int_t^T \dd s \, \del^2 \Big( \frac{1-\ee^{-\kappa (T-s)}}{\kappa} - \frac{1-\ee^{-\kappa (\Tb - s)}}{\kappa} \Big)^2 } 
	 =		\frac{\delta}{\kappa^{3/2}}\sqrt{\frac{\ee^{2 \kappa  T}-\ee^{2 \kappa  t}}{2 (T - t)}} \left(\ee^{-\kappa T}-\ee^{-\kappa  \Tb}\right) .
\end{align}
From the above, it is easy to identify the following limits
\begin{align}
		\lim_{t \to T} \Sig & = {\frac{\del}{\kappa} \left(1 - \ee^{- \kappa  \left(\Tb - T \right)}\right)} , &
    \lim_{T \to \Tb} \Sig & = 0 , &
    \lim_{\Tb \to \infty} \Sig & = {\frac{\delta}{\kappa^{3/2}}\sqrt{\frac{1-\ee^{-2 \kappa  (T-t)}}{2(T-t)}}}, &
    \lim_{t \to T,\Tb \to \infty}\Sig & = \frac{\delta}{\kappa}.
    \label{eq:sig-vasicek-analysis}
\end{align}
In Figure \ref{fig:plot_sigma_vasicek} we plot $\Sig$ as a function of $t$ for various valued of $\Tb$ with $T$ fixed.  

\subsection{Cox-Ingersoll-Ross}
\label{sec:cir}
In the Cox-Ingersoll-Ross (CIR) short-rate model developed in \cite{cox2005theory}, the dynamics of $R=r(Y)$ are given by
\begin{align}
\dd Y_t
	&=	\kappa ( \theta - Y_t) \dd t + \del \sqrt{Y_t}\dd W_t , &
R_t
	&=	Y_t . \label{eq:CIR-model}
\end{align}
Comparing \eqref{eq:CIR-model} with \eqref{eq:R_t} and \eqref{eq:dYi_t}, we see that the functions $r$, $\mu$, and $\sig$ are given by
\begin{align}
r(y)
	&=  y , &
\mu(t,y)
	&=	\kappa ( \theta - y) , &
\sig(t,y)
	&=	\del \sqrt{y} , \label{eq:r-mu-sigma-CIR}
\end{align}
and comparing \eqref{eq:r-mu-sigma-CIR} with \eqref{eq:r-mu-sigma} we identify
\begin{align}
q
	&=	0 , &
\psi
	&=	1 , &
b(t)
	&=	\kappa \theta , &
\beta(t)
	&=	- \kappa , &
\ell(t)
	&=	0, &
\lam(t)
	&=	\del^2 , 
\end{align}
where we have dropped the subscripts from $\psi$, $\beta$ and $\lam$ as $d=1$.
With the above parameters, the solutions $F$ and $G$ of coupled ODEs \eqref{eq:F-ode} and \eqref{eq:G-ode} are
\begin{align}
F(t;T,\nu)
	&=	-\frac{2\kappa \theta}{\del^2}\log \Big(\frac{2\Lam\exp\big((\Lam+\kappa)\tau/2\big)}{-\del^2\nu\big(\exp(\Lam \tau)-1)+\Lam(\exp(\Lam \tau)+1)+\kappa(\exp(\Lam \tau)-1)} \Big) , & 
\tau 
	&:= T-t , \label{eq:F-CIR}\\
G(t;T,\nu)
	&=\frac{2(\exp(\Lam \tau)-1)-\big(\Lam(\exp(\Lam \tau)+1)-\kappa(\exp(\Lam \tau)-1)\big)\nu}{-\del^2\nu\big(\exp(\Lam \tau)-1)+\Lam(\exp(\Lam \tau)+1)+\kappa(\exp(\Lam \tau)-1)}, &
\Lam 
	&:= \sqrt{\kappa^2+2\del^2}. \label{eq:G-CIR}
\end{align}
From \eqref{eq:mu-tilde-sigma-tilde}, \eqref{eq:eta-def}, and \eqref{eq:r-mu-sigma-CIR}, we have
\begin{align}
\sigt(t,x;T,\Tb)	& = \delta\sqrt{\frac{F(t;T,0) - F(t;\Tb,0) - x }{G(t;\Tb,0)-G(t;T,0)}} ,\label{eq:sigma-tilde-CIR}
\end{align}
And thus, using \eqref{eq:A-1d} and \eqref{eq:sigma-tilde-CIR}, the generator $\Act$ is given by
\begin{align}
\Act(t)	
	&=	c(t,x) (\d_x^2 - \d_x ) , &
c(t,x)
	&=	\frac{\del^2}{2}\Big(F(t;T,0) - F(t;\Tb,0)-x\Big)\Big(G(t;\Tb,0) - G(t;T,0)\Big).
\end{align}
Introducing the short-hand notation $c_j(t,x) := \d_x^j c(t,x) / j!$, we have
\begin{align}
    c_{0}(t,x) & = \frac{\del^2}{2}\Big(F(t;T,0) - F(t;\Tb,0)-x\Big)\Big(G(t;\Tb,0) - G(t;T,0)\Big) , \\
    c_{1}(t,x) & \equiv c_1(t) = -\frac{\del^2}{2}\Big(G(t;\Tb,0) - G(t;T,0)\Big) , \\
		c_n(t,x) & = 0 , &
		n		& \geq 2 .
\end{align}
The explicit implied volatility approximation $\Sigb_n$ can now be computed up to order $n=2$ using the formulas in Appendix \ref{sec:explicit-expressions}.  We have
%
%
%
\begin{align}
    \Sig_0 & = \sqrt{\frac{2}{\tau}\int_{t}^T \dd s \, c_0(s,x)} ,
    \\
    \Sig_1 & ={\frac{2(k-x)}{\Sig^3_0\tau^2}\int_{t}^T\dd s \, c_1(s,x)\int_{t}^s \dd q \, c_0(q,x)} , \label{eq:CIR-IV}
    \\
        \Sig_2& = {\frac{6(k-x)^2}{{\Sig_0^7 \tau^4}}\bigg( -2 \Big(\int_t^T  \dd s \, c_1(s) \int_{t}^s  \dd q \, c_0(q,x) \Big){}^2+ \Sig_0^2 \tau \int_t^T  \dd s_1 \, \int_{s_1}^T  \dd s_2 \, c_1(s_1) c_1(s_2) \int_{t}^{s_1}\dd q \,  c_0(q,x) \bigg)} 
    \\
     & \quad {+\frac{(\Sig_0^2 \tau+12)}{{2 \Sig_0^5 \tau^3}}\bigg(\Big(\int_t^T\dd s \, c_1(s) \int_{t}^s  \dd q \, c_0(q,x) \Big){}^2 -\Sig_0^2 \tau \int_t^T  \dd s_1 \, \int_{s_1}^T  \dd  s_2 \, c_1(s_1) c_1(s_2) \int_{t}^{s_1}\dd q \,  c_0(q,x)\bigg)} . 
\end{align}
In Figure \ref{fig:CIR_IV} we plot our explicit approximation of implied volatility $\bar{\Sig}_n$ up to order $n=2$ as a function of $\log$-moneyness $k-x$ with $t=0$ and $\Tb = 2$ fixed and with option maturities ranging over $T = \{\frac{1}{12},\frac{1}{4},\frac{1}{2},\frac{3}{4}\}$.  For comparison, we also plot the exact implied volatility $\Sig$.
We observe that the second order approximation $\Sigb_2$ accurately matches the level, slope, and convexity of the exact implied volatility $\Sig$ near-the-money for all four option maturity dates. In Figure \ref{fig:CIR-Err} we plot the absolute value of the relative error of our second order approximation $|\Sigb_2-\Sig|/\Sig$ as a function of $\log$-moneyness $k-x$ and option maturity $T$. We observe that the error decreases as we approach the origin in both directions of $k-x$ and $T$ and the best approximation region is within $0.2 \%$ of the exact implied volatility.



\subsection{Two-factor Cox-{Ingersoll}-Ross}
\label{sec:cir2}

In the Two-factor Cox-Ingersoll-Ross ({2-D CIR}) short-rate model developed in \cite{cox2005theory}, the dynamics of $R=r(Y)$ are given by
\begin{align}
\dd Y^{(1)}_t
	&=	\kappa_1 ( \theta_1 - Y^{(1)}_t) \dd t + \del_1 \sqrt{Y^{(1)}_t}\dd W^{(1)}_t , &
	\\
	\dd Y^{(2)}_t
	&=	\kappa_2 ( \theta_2 - Y^{(2)}_t) \dd t + \del_2 \sqrt{Y^{(2)}_t}\dd W^{(2)}_t , &
	\\
R_t
	&=	Y^{(1)}_t + Y^{(2)}_t . \label{eq:2D-CIR-model}
\end{align}
Comparing \eqref{eq:2D-CIR-model} with \eqref{eq:R_t} and \eqref{eq:dYi_t}, we see that the functions $r$, $\mu$, and $\sig$ are given by
\begin{align}
r(y_1,y_2)
	&=  y_1+y_2 , &
\mu(t,y_1,y_2)
	&=	\begin{pmatrix}\kappa_1 ( \theta_1 - y_1)
	\\
	\kappa_2 ( \theta_2 - y_2) \end{pmatrix} , &
\sig(t,y_1,y_2)
	&=	\begin{pmatrix}\del_1 \sqrt{y_1} & 0 
	\\
	0 & \del_2 \sqrt{y_2}  \end{pmatrix} , \label{eq:r-mu-sigma-2D-CIR}
\end{align}
and comparing \eqref{eq:r-mu-sigma-2D-CIR} with \eqref{eq:r-mu-sigma} we identify
\begin{align}
q
	&=	0 , &
\psi
	&=	\begin{pmatrix}1
	\\
	1 \end{pmatrix} , &
b(t)
	&=	\begin{pmatrix}\kappa_1 \theta_1
	\\
	\kappa_2 \theta_2  \end{pmatrix} , &
\beta_1(t)
	&=	- 	\begin{pmatrix}\kappa_1 
	\\
	0 \end{pmatrix} , &
	\\
	\beta_2(t)
	&=	- 	\begin{pmatrix} 0
	\\
	\kappa_2 \end{pmatrix} , &
\ell(t)
	&=	0, &
\lam_1(t)
	&=	\begin{pmatrix}\del^2_1  & 0 
	\\
	0 & 0  \end{pmatrix} ,  &
\lam_2(t)
	&=	\begin{pmatrix} 0 & 0 
	\\
	0 & \del^2_2  \end{pmatrix} . 
\end{align}
With the above parameters, the solutions $F$ and $G = (G_1,G_2)$ of coupled ODEs \eqref{eq:F-ode} and \eqref{eq:G-ode} are
\begin{align}
F(t;T,\nu)
	&=-\sum_{i=1}^2 \frac{2\kappa_i \theta_i}{\del_i^2}\log \Big(\frac{2\Lam_i\exp\big((\Lam_i+\kappa_i)\tau/2\big)}{-\del_i^2\nu_i\big(\exp(\Lam_i \tau)-1)+\Lam_i(\exp(\Lam_i \tau)+1)+\kappa_i(\exp(\Lam_i \tau)-1)} \Big), \label{eq:F-2d-CIR}\\
G_i(t;T,\nu)
    &= \frac{2(\exp(\Lam_i \tau)-1)-\big(\Lam_i(\exp(\Lam_i \tau)+1)-\kappa_i(\exp(\Lam_i \tau)-1)\big)\nu_i}{-\del_i^2\nu_i\big(\exp(\Lam_i \tau)-1)+\Lam_i(\exp(\Lam_i \tau)+1)+\kappa_i(\exp(\Lam_i \tau)-1)},  & i =\{1,2\} ,
    \label{eq:G-2d-CIR}
    \\
    \Lam_i & := \sqrt{\kappa_i^2+2\del_i^2}.
\end{align}
From \eqref{eq:mu-tilde-sigma-tilde}, \eqref{eq:eta-def}, and \eqref{eq:r-mu-sigma-2D-CIR}, we have
\begin{align}
    \eta(t,x,y_2;T,\Tb) & = 	\frac{F(t;T,0) - F(t;\Tb,0) - x  + \Big(G_2(t;T,0) - G_2(t;\Tb,0) \Big) y_2}{G_1(t;\Tb,0)-G_1(t;T,0)},
    \\
\sigt(t,x,y_2;T,\Tb)	& = \begin{pmatrix} \del_1\sqrt{\eta(t,x,y_2;T,\Tb)} &  0 \\ 0 & \del_2 \sqrt{y_2} \end{pmatrix}
,\label{eq:sigma-tilde-2D-CIR}
\end{align}
and thus, using \eqref{eq:A-2d} and \eqref{eq:sigma-tilde-2D-CIR}, the generator $\Act$ is given by
\begin{align}
\Act(t)	
	&=	c(t,x,y_2) (\d_x^2 - \d_x ) + f(t,x,y_2) \d_{y_2} + g(t,x,y_2) \d_{y_2}^2  + h(t,x,y_2) \d_x \d_{y_2},
\end{align}
where the functions $c$, $f$, $g$ and $h$ are given by 
\begin{align}
c(t,x,y_2)
	&=	 \tfrac{1}{2}\del^2_1 \bigg(F(t;T,0) - F(t;\Tb,0) - x  + \Big(G_2(t;T,0) - G_2(t;\Tb,0) \Big) y_2\bigg)\Big(G_1(t;\Tb,0)-G_1(t;T,0)\Big)
	\\
	& \quad + \tfrac{1}{2} \del^2_2  \Big(G_2(t;T,0) - G_2(t;\Tb,0)\Big)^2 y_2 ,
	\\
f(t,x,y_2)
	&=	\kappa_2(\theta_2-y_2) -\del^2_2 y_2 G_2(t;T,0), \\ 
g(t,x,y_2)
	&=	\tfrac{1}{2}\del^2_2 y_2, \\
h(t,x,y_2)
	&=\del^2_2 y_2 \Big( G_2(t;T,0) - G_2(t;\Tb,0) \Big).
\end{align}
Introducing the notation $\chi_{i,j}(t,x,y_2) := \d_x^i \d_{y_2}^j \chi(t,x,y_2)/ (i! j!)$ where $ \chi \in \{c,f,g,h\}$, we compute
\begin{align}
    \chi_{0,0}(t,x,y_2) & = \chi(t,x,y_2),
    \\
	c_{1,0}(t,x,y_2) & = -\tfrac{1}{2}\del^2_1 \Big(G_1(t;\Tb,0)-G_1(t;T,0) \Big),
	\\
	c_{0,1}(t,x,y_2) & = \tfrac{1}{2} \del^2_1 \Big( G_2(t;T,0) - G_2(t;\Tb,0) \Big) \Big( G_1(t;\Tb,0)-G_1(t;T,0) \Big) \\ & \quad
		+ \tfrac{1}{2} \del^2_2  \Big( G_2(t;T,0) - G_2(t;\Tb,0) \Big)^2 ,
		\\
		f_{0,1}(t,x,y_2) & =  -(\kappa_2 +\del^2_2)G_2(t;T,0),
		\\
	g_{0,1}(t,x,y_2)
	&=	\tfrac{1}{2}\del^2_2,
	\\
	h_{0,1}(t,x,y_2)
	&= \del^2_2 \Big( G_2(t;T,0) - G_2(t;\Tb,0) \Big) ,
\end{align}
and $\chi_{i,j}(t,x,y_2)  = 0$, for any term not given above. The explicit implied volatility approximation $\Sigb_n$ can now be computed up to order $n=2$ using the formulas in Appendix \ref{sec:explicit-expressions}. We have
\begin{align}
    \Sig_0 & = \sqrt{\frac{2}{\tau}\int_{t}^T \dd s \, c_{0,0}(s,x,y_2)},
\\
    \Sig_1 & = \frac{(k-x)}{\tau^2\Sig^3_0}\Big(2\int_{t}^T\dd s \, c_{1,0}(s,x,y_2)\int_{t}^s \dd q \, c_{0,0}(q,x,y_2)+ \int_{t}^T \dd s \, c_{0,1}(s,x,y_2)\int_{t}^s \dd q \, h_{0,0}(q,x,y_2)\Big)
    \\
     &\quad + \frac{1}{2\tau\Sig_0}\int_{t}^T \dd s \, c_{0,1}(s,x,y_2)\Big(2\int_{t}^s \dd q \, f_{0,0}(q,x,y_2)+ \int_{t}^s \dd q \, h_{0,0}(q,x,y_2)\Big) , \label{eq:2D-CIR-IV}
\end{align}
where we have omitted the 2nd order term $\Sig_2$ due to its considerable length.
\\[0.5em]
In Figure \ref{fig:2D-CIR-IV} we plot our explicit approximation of implied volatility $\bar{\Sig}_n$ up to order $n=2$ as a function of $\log$-moneyness $k-x$ with $t=0$ and $\Tb = 2$ fixed and with option maturities ranging over $T = \{\frac{1}{12},\frac{1}{4},\frac{1}{2},\frac{3}{4}\}$.  For comparison, we also plot the the exact implied volatility $\Sig$.
As is the case with the (1-D) CIR model, we observe in the 2-D CIR model that the second order approximation $\Sigb_2$ accurately matches the level, slope, and convexity of the exact implied volatility $\Sig$ near-the-money for all four option maturity dates. In Figure \ref{fig:2D-CIR-Err} we plot the absolute value of the relative error of our second order approximation $|\Sigb_2-\Sig|/\Sig$ as a function of $\log$-moneyness $k-x$ and option maturity $T$. We observe that the error decreases as we approach the origin in both directions of $k-x$ and $T$ and the best approximation region is within $0.1 \%$ of the exact implied volatility.

\subsection{Fong-Vasicek}
\label{sec:fong}
In the Fong-Vasicek short-rate model developed in \cite{fong1991fixed}, the dynamics of $R=r(Y)$ are given by
\begin{align}
\dd Y^{(1)}_t
	&=	\kappa_1 ( \theta_1 - Y^{(1)}_t) \dd t + \sqrt{Y^{(2)}_t}\dd W^{(1)}_t , &
	\\
	\dd Y^{(2)}_t
	&=	\kappa_2 ( \theta_2 - Y^{(2)}_t) \dd t + \del_2 \rho \sqrt{Y^{(2)}_t} \dd W^{(1)}_t + \del_2 \rhob \sqrt{Y^{(2)}_t} \dd W^{(2)}_t , & \quad & \rhob = \sqrt{1-\rho^2}
	\\
R_t
	&=	Y^{(1)}_t. \label{eq:FV-model}
\end{align}
Comparing \eqref{eq:FV-model} with \eqref{eq:R_t} and \eqref{eq:dYi_t}, we see that the functions $r$, $\mu$, and $\sig$ are given by
\begin{align}
r(y_1,y_2)
	&=  y_1 , &
\mu(t,y_1,y_2)
	&=	\begin{pmatrix}\kappa_1 ( \theta_1 - y_1)
	\\
	\kappa_2 ( \theta_2 - y_2) \end{pmatrix} , &
\sig(t,y_1,y_2)
	&=	\begin{pmatrix} \sqrt{y_2} & 0 
	\\
	\del_2 \rho \sqrt{y_2} & \del_2 \rhob \sqrt{y_2}  \end{pmatrix} , \label{eq:r-mu-sigma-FV}
\end{align}
and comparing \eqref{eq:r-mu-sigma-FV} with \eqref{eq:r-mu-sigma} we identify
\begin{align}
q
	&=	0 , &
\psi
	&=	\begin{pmatrix}1
	\\
	0 \end{pmatrix} , &
b(t)
	&=	\begin{pmatrix}\kappa_1 \theta_1
	\\
	\kappa_2 \theta_2  \end{pmatrix} , &
\beta_1(t)
	&=	- 	\begin{pmatrix}\kappa_1 
	\\
	0 \end{pmatrix} , &
	\\
	\beta_2(t)
	&=	- 	\begin{pmatrix} 0
	\\
	\kappa_2 \end{pmatrix} , &
\ell(t)
	&=	0, &
\lam_1(t)
	&=	\begin{pmatrix}0  & 0 
	\\
	0 & 0  \end{pmatrix} ,  &
\lam_2(t)
	&=	\begin{pmatrix} 1 & \delta_2\rho
	\\
	\delta_2\rho & \del^2_2  \end{pmatrix} . 
\end{align}
With the above parameters, we find using \eqref{eq:F-ode} and \eqref{eq:G-ode} that the ODEs satisfied by $F$ and $G = (G_1, G_2)$ are 
\begin{align}
\d_t F(t;T,\nu)
	&=	-\kappa_1\theta_1 G_1(t;T,\nu)-\kappa_2\theta_2 G_2(t;T,\nu) , &
F(T;T,\nu)
	&=	0 , \label{eq:F-ode-FV} \\
\d_t G_1(t;T,\nu)
	&=	\kappa_1 G_1(t;T,\nu)-1 , &
G_1(T;T,\nu)
	&=	- \nu_1 , \label{eq:G1-ode-FV} \\
	\d_t G_2(t;T,\nu)
	&=	\tfrac{1}{2}\del^2_2 G^2_2(t;T,\nu) + \Big(\del_2\rho G_1(t;T,\nu) + \kappa_2 \Big) G_2(t;T,\nu) 
	\\
	& \quad + \tfrac{1}{2}G^2_1(t;T,\nu) , &
G_2(T;T,\nu)
	&=	- \nu_2 . \label{eq:G2-ode-FV}
\end{align}
Although one can obtain explicit expressions for $F(t;T,\nu)$, $G_1(t;T,\nu)$ and $G_2(t;T,\nu)$, these expressions are given in terms of \textit{confluent hypergeometric fuctions} (CHFs).  As numerical evaluation of CHFs is time-consuming, computing explicit Call prices using \eqref{eq:v-exact} is not practical because it involves integrals with respect to $\nu$.
By contrast, in order to compute our explicit approximation of implied volatility $\Sigb_n$, we need only expressions for $F(t;T,0)$, $G_1(t;T,0)$ and $G_2(t;T,0)$, which we provide in Appendix \ref{sec:F-and-G}.
\\[0.5em]
From \eqref{eq:mu-tilde-sigma-tilde}, \eqref{eq:eta-def}, and \eqref{eq:r-mu-sigma-FV}, we have
\begin{align}
    \eta(t,x,y_2;T,\Tb) & = 	\frac{F(t;T,0) - F(t;\Tb,0) - x  + \Big(G_2(t;T,0) - G_2(t;\Tb,0)\Big) y_2}{G_1(t;\Tb,0)-G_1(t;T,0)},
    \\
\sigt(t,x,y_2;T,\Tb)& = \begin{pmatrix} \sqrt{y_2} & 0 
	\\
	\del_2 \rho \sqrt{y_2} & \del_2 \rhob \sqrt{y_2}  \end{pmatrix}
,\label{eq:sigma-tilde-FV}
\end{align}
And thus, using \eqref{eq:A-2d} and \eqref{eq:sigma-tilde-FV}, the generator $\Act$ is given by
\begin{align}
\Act(t)	
	&=	c(t,x,y_2) (\d_x^2 - \d_x ) + f(t,x,y_2) \d_{y_2} + g(t,x,y_2) \d_{y_2}^2  + h(t,x,y_2) \d_x \d_{y_2},
\end{align}
where the functions $c$, $f$, $g$ and $h$ are given by 
\begin{align}
c(t,x,y_2)
	&=	 \tfrac{1}{2}y_2 \Big(G_1(t;T,0) - G_1(t;\Tb,0)\Big)^2 + \rho\del_2  y_2 \Big(G_1(t;T,0) - G_1(t;\Tb,0)\Big) \Big(G_2(t;T,0) - G_2(t;\Tb,0)\Big)
	\\
	& \quad + \tfrac{1}{2}\del^2_2 y_2 \Big(G_2(t;T,0) - G_2(t;\Tb,0)\Big)^2,
	\\
f(t,x,y_2)
	&=	\kappa_2(\theta_2-y_2) -\del^2_2 y_2 G_2(t;T,0) - \rho \del_2 y_2 G_1(t;T,0),  \\ 
g(t,x,y_2)
	&=	\tfrac{1}{2}\del^2_2 y_2, \\
h(t,x,y_2)
	&= \del^2_2 y_2 \Big( G_2(t;T,0) - G_2(t;\Tb,0)\Big) + \rho \del_2 y_2 \Big(G_1(t;T,0) - G_1(t;\Tb,0)\Big).
\end{align}
Once again using the short-hand notation $\chi_{i,j}(t,x,y_2) := \d_x^i \d_{y_2}^j \chi(t,x,y_2)/ (i! j!)$ where $\chi \in \{c,f,g,h\}$, we compute
\begin{align}
    \chi_{0,0}(t,x,y_2) & = \chi(t,x,y_2),
	\\
	c_{0,1}(t,x,y_2) & =  \tfrac{1}{2} \Big(G_1(t;T,0) - G_1(t;\Tb,0)\Big)^2 + \rho\del_2 \Big(G_1(t;T,0) - G_1(t;\Tb,0)\Big) \Big(G_2(t;T,0) - G_2(t;\Tb,0) \Big)
	\\
	&\quad  + \tfrac{1}{2}\del^2_2 \Big(G_2(t;T,0) - G_2(t;\Tb,0) \Big)^2,  
		\\
		f_{0,1}(t,x,y_2) & =  -\kappa_2 -\del^2_2 G_2(t;T,0) - \rho \del_2 G_1(t;T,0),
	\\
	g_{0,1}(t,x,y_2)
	&=	\tfrac{1}{2}\del^2_2,
	\\
	h_{0,1}(t,x,y_2)
	&= \del^2_2 \Big( G_2(t;T,0) - G_2(t;\Tb,0) \Big) + \rho \del_2 \Big( G_1(t;T,0) - G_1(t;\Tb,0) \Big) ,
\end{align} 
where $\chi_{i,j}(t,x,y_2)  = 0$ for any term not given above. The explicit implied volatility approximation $\Sigb_n$ can now be computed up to order $n=2$ using the formulas in Appendix \ref{sec:explicit-expressions}.  We have
\begin{align}
    \Sig_0 & = \sqrt{\frac{2}{\tau}\int_{t}^T \dd s \, c_{0,0}(s,x,y_2)},
\\
    \Sig_1 & = \frac{k-x}{\tau^2\Sig^3_0}\Big(\int_{t}^T \dd s \, c_{0,1}(s,x,y_2)\int_{t}^s \dd q \, h_{0,0}(q,x,y_2)\Big)
    \\
    & \quad + \frac{1}{2\tau\Sig_0}\int_{t}^T \dd s \, c_{0,1}(s,x,y_2)\Big(2\int_{t}^s \dd q \, f_{0,0}(q,x,y_2)+ \int_{t}^s \dd q \, h_{0,0}(q,x,y_2)\Big). \label{eq:FV-IV}
\end{align}
where we have omitted the second order term $\Sig_2$ due to its considerable length.
\\[0.5em]
In Figure \ref{fig:plot_sigma_FV} we plot our second order approximation of implied volatility $\Sigb_2$ as a function of $\log$-moneyness $k-x$ with the maturity date of the bond fixed at $\Tb = 2$, the maturity date of the option taking the following values $T = \{\frac{1}{12}, \frac{1}{4}, \frac{1}{2}, \frac{3}{4}\}$ and the correlation parameter taking the following values $\rho = \{-0.7,-0.3,0.3,0.7\}$. We can see the convexity near-the-money changes from concave to convex as we increase $\rho$. From the expression of $\Sig_1$ in \eqref{eq:FV-IV} we observe that  the slope of $\Sig_1$ with respect to $k-x$ is controlled by the sign of $c_{0,1}$ and $h_{0,0}$. As $G(t;T,0)$ is an increasing function in $T$, the expression $G_i(t;T,0)-G_i(t;\Tb,0)$ is negative, which means that, fixing all other parameters, $\rho$ controls the sign of $c_{0,1}$ and $h_{0,0}$. As a result, as we change $\rho$ from $-1$ to $1$ the slope of $\Sig_1$ changes accordingly.  A similar analysis can be done on the sign of coefficients of $(k-x)^2$ of $\Sig_2$ to show that $\rho$ controls the convexity of $\Sig_2$ with respect to $k-x$. This is in contrast to the CIR and {2-D CIR} models, where the implied volatility curve near-the-money is concave.

\section{Conclusion} 
\label{sec:conclusion}
In this paper, we have provided an explicit asymptotic approximation for the implied volatility of Call options on bonds assuming the {short-rate} is given by an affine term-structure model.  In future work, we plan to extend our results by providing explicit implied volatility approximations {for other} short-rate derivatives including caps and floors.

%
%

\appendix

\section{Explicit expressions for $\Sig_0$, $\Sig_1$ and $\Sig_2$} 
\label{sec:explicit-expressions}
In this appendix we give the expressions for the implied volatility approximation using \eqref{eq:sig-0} and \eqref{eq:sig-n} explicitly up to second order for $d=\{1,2\}$ in terms of the coefficients $c$, $f$,$g$, and $h$ of $\Act$, given in \eqref{eq:A-2d}, by performing Taylor's expansion of the coefficients around $\zb(t) = (x,\yt)$.  To ease the notation, we define
\begin{align}
 \chi_{i,j}(t) \equiv \chi_{i,j}(t,x,\yt) &   = \frac{\d^i_x\d^j_y\chi(t,x,\yt)}{i!j!}, &  \chi & \in \{c,f,g,h\}. \label{eq:chi-ij}
\end{align}
The zeroth order term $\Sig_0$ is given by
\begin{align}
        \Sig_0 &  = \sqrt{\frac{2}{\tau}\int_{t}^T \dd s \, c_{0,0}(s,x,\yt) } .
\end{align}
Next, let {us} define 
\begin{align}
H_n(\xi) &:= \Big(\frac{-1}{\Sigma_0\sqrt{2\tau}}\Big)^n \mathscr{H}_n(\xi), & 
\xi &:= \frac{x-k-\frac{1}{2}\Sigma^2_0\tau}{\Sig_0\sqrt{2\tau}} , &
\tau &:= T-t,
\end{align}
where $\mathscr{H}_n(\xi)$ is the  $n$th-order \emph{Hermite's polynomial}.
Then the first order term $\Sig_1$ is given by
\begin{align}
\Sig_1 & = \Sig_{1,0}+\Sig_{0,1},  
\end{align}
where $ \Sig_{1,0}$ and $\Sig_{0,1}$ are given by
\begin{align}
    \Sig_{1,0} & = \frac{1}{\tau\Sig_0}\int_{t}^T\dd s \, c_{1,0}(s,x,\yt)\int_{t}^s \dd q \, c_{0,0}(q,x,\yt)  \Big(2H_1(\xi)-1\Big),  \\
    \Sig_{0,1} & = \frac{1}{\tau\Sig_0}\int_{t}^T \dd s \, c_{0,1}(s,x,\yt)\Big(\int_{t}^s \dd q \, f_{0,0}(q,x,\yt)+ \int_{t}^s \dd q \, h_{0,0}(q,x,\yt) H_1(\xi)\Big) .
\end{align}
Note that $\Sig_{0,1} = 0$ when $d=1$ because in this case $f=h=0$.  Lastly, the second order term $\Sig_2$ is given by
\begin{align}
\Sig_{2}
    & = \Sig_{2,0}+\Sig_{1,1}+\Sig_{0,2}, 
\end{align}
where, using the short-hand notation $\xi_{i,j}(t) := \xi_{i,j}(t,x,\yt)$, the terms $\Sig_{2,0}$, $\Sig_{1,1}$, $\Sig_{0,2}$ are given by
\begin{align}
    \Sig_{2,0} & = \frac{1}{\tau\Sig_0}\bigg(\frac{1}{2}\int_{t}^T \dd s \, c_{2,0}(s) \Big((\int_{t}^s \dd q \, c_{0,0}(q) )^2 (4H_2(\xi)-4H_1(\xi)+1) + 2 \int_{t}^s \dd q \, c_{0,0}(q)  \Big)
    \\
    & \quad + \int_{t}^T \dd s_1 \, \int_{s_1}^T \dd s_2 \, c_{1,0}(s_1)c_{1,0}(s_2)\Big(\int_{t}^{s_1} \dd q_1 \, c_{0,0}(q_1) \int_{t}^{s_2} \dd q_2 \, c_{0,0}(q_2) 
    \big(4H_4(\xi)-8H_3(\xi)+5H_2(\xi)-H_1(\xi)\big) 
    \\
    & \quad + \int_{t}^{s_1} \dd q_1 \, c_{0,0}(q_1)  \Big(6H_2(\xi)-6H_1(\xi)+1)\Big)\bigg) -\frac{\Sig^2_{1,0}}{2}\Big(\tau\Sig_0 (H_2(\xi)-H_1(\xi))
    + \frac{1}{\Sig_0} \Big), \\
    \Sig_{1,1} & = \frac{1}{\tau\Sig_0}\Bigg(\frac{1}{2}\int_{t}^T \dd s \, c_{1,1}(s) \bigg(2\int_{t}^{s} \dd q_1 \, c_{0,0}(q_1) \int_{t}^{s} \dd q_2 \, h_{0,0}(q_2) H_2(\xi) 
    \\
    & \quad + \int_{t}^{s} \dd q_1 \, c_{0,0}(q_1)(2\int_{t}^{s} \dd q_2 \, f_{0,0}(q_2)-\int_{t}^{s} \dd q_2 \, h_{0,0}(q_2))H_1(\xi) -\int_{t}^{s} \dd q_1 \, c_{0,0}(q_1)\int_{t}^{s} \dd q_2 \, f_{0,0}(q_2)+\int_{t}^{s} \dd q_1 \, h_{0,0}(q_1)\bigg) 
    \\
    & \quad + \int_{t}^T \dd s_1 \,\int_{s_1}^{T} \dd s_2 \, c_{1,0}(s_1)c_{0,1}(s_2) \bigg( 2\int_{t}^{s_1} \dd q_1 \, c_{0,0}(q_1)\int_{t}^{s_2} \dd q_2 \, h_{0,0}(q_2)H_4(\xi) 
        \\
        & \quad +\int_{t}^{s_1} \dd q_1 \, c_{0,0}(q_1) \Big(2\int_{t}^{s_2} \dd q_2 \, f_{0,0}(q_2) - 3\int_{t}^{s_2} \dd q_2 \, h_{0,0}(q_2)\Big)H_3(\xi) \\
        & \quad +(\int_{t}^{s_1} \dd q \, c_{0,0}(q)(\int_{t}^{s_2} \dd q \, h_{0,0}(q)-3 \int_{t}^{s_2} \dd q \, f_{0,0}(q))+\int_{t}^{s_1} \dd q \, h_{0,0}(q)\Big) H_2(\xi) 
        \\
        & \quad + \Big(\int_{t}^{s_1} \dd q_1 \, c_{0,0}(q_1)\int_{t}^{s_2} \dd q_2 \, f_{0,0}(q_2)-\int_{t}^{s_1} \dd q_1 \, h_{0,0}(q_1)\Big)H_1(\xi) \bigg)
        \\
        & \quad + \int_{t}^T \dd s_1 \,\int_{s_1}^{T} \dd s_2 \, c_{0,1}(s_1)c_{1,0}(s_2) \bigg(2 \int_{t}^{s_1} \dd q_1 \, h_{0,0}(q_1) \int_{t}^{s_2} \dd q_2 \, c_{0,0}(q_2) H_4(\xi) 
        \\
        & \quad +  \Big(2\int_{t}^{s_1} \dd q_1 \, f_{0,0}(q_1)-3\int_{t}^{s_1} \dd q_1 \, h_{0,0}(q_1)\Big)\int_{t}^{s_2} \dd q_2 \, c_{0,0}(q_2)H_3(\xi) 
        \\
        & \quad + \Big(\big(\int_{t}^{s_1} \dd q_1 \, h_{0,0}(q_1)-3 \int_{t}^{s_1} \dd q_1 \, f_{0,0}(q_1)\big)\int_{t}^{s_2} \dd q_2 \, c_{0,0}(q_2) + 3 \int_{t}^{s_1} \dd q_1 \, h_{0,0}(q_1)\Big)H_2(\xi) 
        \\
        & \quad + \Big(\int_{t}^{s_1} \dd q_1 \, f_{0,0}(q_1)(2 +\int_{t}^{s_2} \dd q_2 \, c_{0,0}(q_2))-2\int_{t}^{s_1} \dd q_1 \, h_{0,0}(q_1)\Big) H_1(\xi) - \int_{t}^{s_1} \dd q_1 \, f_{0,0}(q_1) \bigg)
        \\
        & \quad + \int_{t}^T \dd s_1 \,\int_{s_1}^{T} \dd s_2 \,f_{1,0}(s_1)c_{0,1}(s_2)\int_{t}^{s_1} \dd q_1 \, c_{0,0}(q_1)\Big(2H_1(\xi) -1\Big)
        \\
        & \quad + 2\int_{t}^T \dd s_1 \,\int_{s_1}^{T} \dd s_2 \,h_{1,0}(s_1)c_{0,1}(s_2)\int_{t}^{s_1} \dd q_1 \, c_{0,0}(q_1)\Big(2H_2(\xi)  -H_1(\xi)\Big) \Bigg)
        \\
        & \quad -\Sig_{1,0}\Sig_{0,1}\Big(\tau\Sig_0 (H_2(\xi)-H_1(\xi))
    + \frac{1}{\Sig_0} \Big), \\
    \Sig_{0,2} & = \frac{1}{\tau\Sig_0}\Bigg(\frac{1}{2}\int_{t}^T \dd s \, c_{0,2}(s) \bigg(\Big(\int_{t}^{s} \dd q \, h_{0,0}(q)\big)^2 H_2(\xi) + 2\int_{t}^{s} \dd q_1 \, h_{0,0}(q_1)\int_{t}^{s} \dd q_2 \, f_{0,0}(q_2)H_1(\xi) 
    \\
    & \quad + (\int_{t}^{s} \dd q \, f_{0,0}(q))^2 + 2\int_{t}^{s} \dd q \, g_{0,0}(q)\bigg)
    \\
    & \quad + \int_{t}^T \dd s_1 \,\int_{s_1}^{T} \dd s_2 \, c_{0,1}(s_1)c_{0,1}(s_2) \bigg( \int_{t}^{s_1} \dd q_1 \, h_{0,0}(q_1) \int_{t}^{s_2} \dd q_2 \, h_{0,0}(q_2) H_4(\xi)
    \\
    & \quad +  \Big(\int_{t}^{s_1} \dd q_1 \, f_{0,0}(q_1) \int_{t}^{s_2} \dd q_2 \, h_{0,0}(q_2) + \int_{t}^{s_1} \dd q_1 \, h_{0,0}(q_1)\int_{t}^{s_2} \dd q_2 \, f_{0,0}(q_2) -\int_{t}^{s_1} \dd q_1 \, h_{0,0}(q_1) \int_{t}^{s_2} \dd q_2 \, h_{0,0}(q_2) \Big)  H_3(\xi)
    \\
    & \quad + \Big(2\int_{t}^{s_1} \dd q \, g_{0,0}(q) + \int_{t}^{s_1} \dd q_1 \, f_{0,0}(q_1) \int_{t}^{s_2} \dd q_2 \, f_{0,0}(q_2) -\int_{t}^{s_1} \dd q_1 \, f_{0,0}(q_1) \int_{t}^{s_2} \dd q_2 \, h_{0,0}(q_2)
    \\
    & \quad - \int_{t}^{s_2} \dd q_1 \, f_{0,0}(q_1) \int_{t}^{s_1} \dd q_2 \, h_{0,0}(q_2)\Big)H_2(\xi) -\Big(2\int_{t}^{s_1} \dd q \, g_{0,0}(q)+\int_{t}^{s_1} \dd q_1 \, f_{0,0}(q_1) \int_{t}^{s_2} \dd q_2 \, f_{0,0}(q_2)\Big) H_1(\xi)\bigg)
    \\
    & \quad + \int_{t}^T \dd s_1 \,\int_{s_1}^{T} \dd s_2 \, f_{0,1}(s_1)c_{0,1}(s_2)\Big(\int_{t}^{s_1} \dd q \, h_{0,0}(q)H_1(\xi)+\int_{t}^{s_1} \dd q \, f_{0,0}(q)\Big)
    \\
    & \quad + \int_{t}^T \dd s_1 \,\int_{s_1}^{T} \dd s_2 \, h_{0,1}(s_1)c_{0,1}(s_2)\Big(\int_{t}^{s_1} \dd q \, h_{0,0}(q) H_2(\xi) + \int_{t}^{s_1} \dd q \, f_{0,0}(q)H_1(\xi)\Big)\Bigg)
    \\
    & \quad -\frac{\Sig^2_{0,1}}{2}\Big(\tau\Sig_0 (H_2(\xi)-H_1(\xi))
    + \frac{1}{\Sig_0} \Big).
\end{align}
Note that when $d=1$ we have that $\Sig_{1,1} = \Sig_{0,2} = 0$ because in this case $f = g = h = 0$.

\section{Expressions for $F$, $G_1$ and $G_2$ in the Fong-Vasicek setting} 
\label{sec:F-and-G}
We can derive from \eqref{eq:F-ode-FV} and \eqref{eq:G1-ode-FV} that
\begin{align}
    F(t;T,0) & = \kappa_1\theta_1 \int_{t}^T  \dd s \, G_1(s;T,0) +\kappa_2\theta_2\int_{t}^T \dd s \, G_2(s;T,0),
    \\
    G_1(t;T,0) & = \frac{1-\ee^{-\kappa_1(T-t)}}{\kappa_1},
\end{align}
and from $\eqref{eq:G2-ode-FV}$ that
\begin{align}
G_2(t;T,0)  & = \frac{\ee^{-\kappa_1(T-t)}}{\del_2^2 \kappa_1^3} \left(\left(\alphab_1+\alphab_2 \ee^{\kappa_1 (T-t)}\right) + \frac{\betab \lamb U\left(\Phib+1,\Psib+1,\ee^{- \kappa_1 (T-t)} \zetab\right)+\gamb M\left(\Phib+1,\Psib+1,\ee^{- \kappa_1 (T-t)} \zetab\right)}{\lamb U\left(\Phib,\Psib,\ee^{- \kappa_1 (T-t)} \zetab\right)+ M\left(\Phib,\Psib,\ee^{- \kappa_1 (T-t)} \zetab\right)}\right),
\end{align}
where we have introduced constants
\begin{align}
    \alphab & = \alphab_1+\alphab_2, 
    & 
    \alphab_1 & = \del_2 \kappa_1^2 (\rho +i \rhob), 
    &  
    \alphab_2 & = -\kappa_1^2 (\del_2 \rho +\kappa_1 \kappa_2 +\betab_2),
    \\
    \betab & = \del_2 \left(\betab_1+i \rhob (\betab_2+\kappa_1^2)\right),
    & 
    \betab_1 & = \del_2 \rhob^{2} + \rho \kappa_1(\kappa_1-\kappa_2),
    &
    \betab_2 & = \sqrt{(\del_2 \rho +\kappa_1 \kappa_2)^2-\del_2^2},
    \\
    \Phib & = \frac{\Psib}{2}+\frac{\betab_1}{2 i \kappa_1^2 \rhob},
    & 
    \Psib & = \frac{\betab_2}{\kappa_1^2}+1, &
    \zetab & = \frac{i \del_2 \rhob}{\kappa_1^2},
    \\
    \lamb & = -\frac{\gamb M(\Phib+1,\Psib+1,\zetab) + \alphab M(\Phib,\Psib,\zetab)}{\betab U(\Phib+1,\Psib+1,\zetab)+\alphab U(\Phib,\Psib,\zetab)}, 
    & 
    \gamb & = -\frac{2 \Phib \kappa_1^4 \zetab}{\Psib}. 
    &
\end{align}
and where $M$ and $U$ are CHFs of the first kind and second kind, respectively.  Explicitly, we have
\begin{align}
    M(a,b,z) & = \sum_{n=0}^{\infty} \frac{a(a+1)\ldots (a+n)}{b(b+1)\ldots (b+n)} \frac{z^n}{n!},
    \\
    U(a,b,z) & = \frac{\Gam_E (1-b)}{\Gam_E (a+1-b)}M(a,b,z)+\frac{\Gam_E (b-1)}{\Gam_E (a)}z^{1-b}M(a+1-b,2-b,z),
\end{align}
where $\Gam_E$ is the \emph{Euler Gamma function}.

%
%

\bibliography{bibliography}

%
%

\clearpage

\begin{figure}
     \centering
    \includegraphics[width=0.975\textwidth]{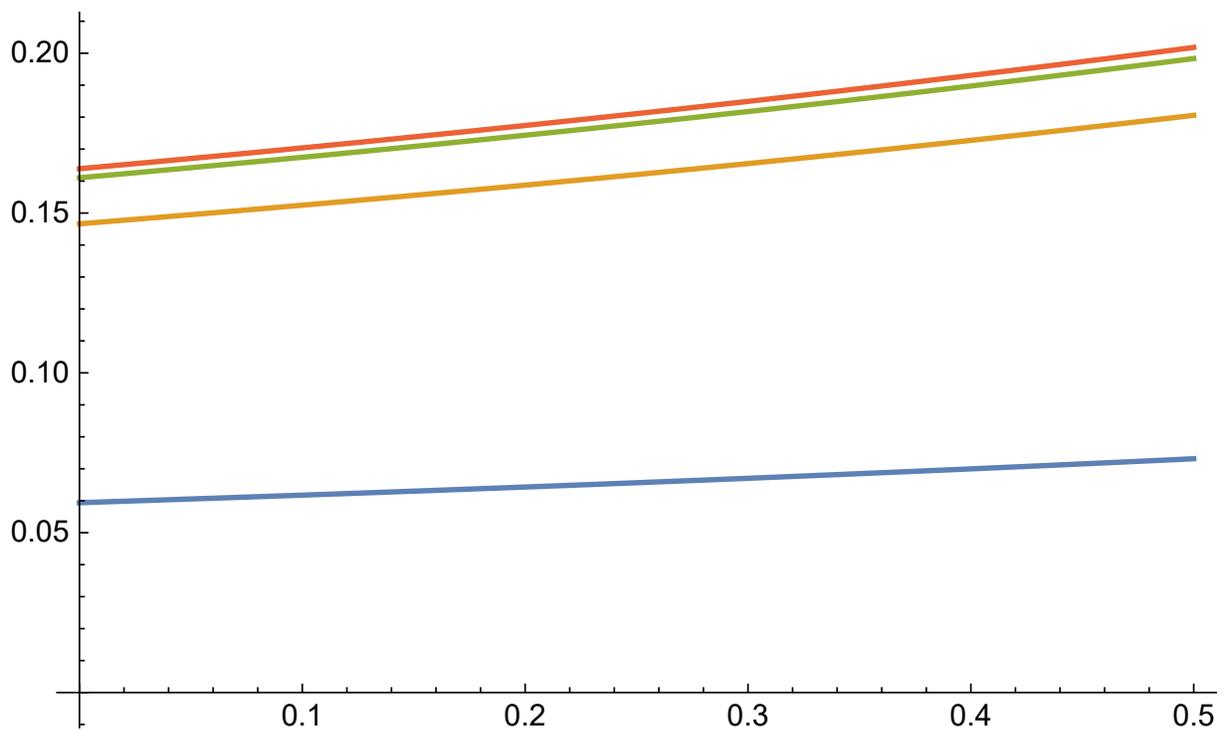}\\
    \caption{For the Vasicek short-rate model described in Section \ref{sec:vasicek}, we plot implied volatility $\Sig$ as a function of $t$ with the maturity date of the options fixed at $T = 0.5$ and with the maturity date of the underlying bond taking the following values $\Tb = \{1, 3, 5, 10\}$, which correspond to the blue, orange, green, and red curves, respectively.  The following model parameters remained fixed: $\kappa = 0.9$, $\delta = \sqrt{0.033}$, and $\theta = \frac{0.08}{0.9}$.
		}
    \label{fig:plot_sigma_vasicek}
\end{figure}

\begin{figure}
\centering
\begin{tabular}{cc}
\includegraphics[width=0.45\textwidth]{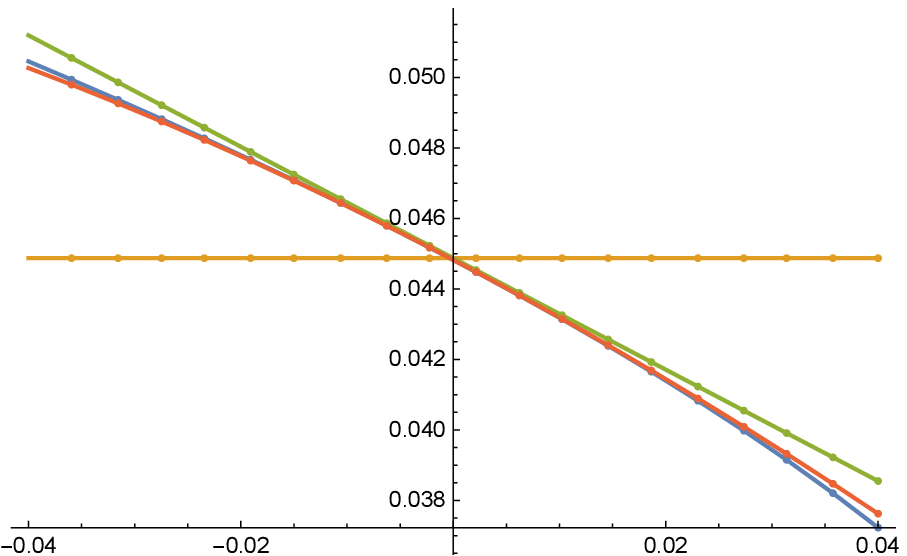}&
\includegraphics[width=0.45\textwidth]{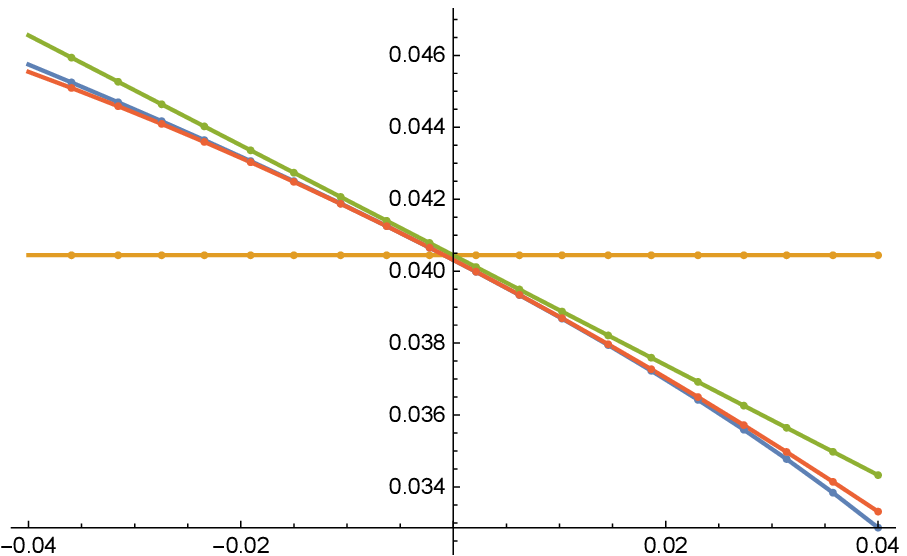}\\
$T = \frac{1}{12}$ & $T = \frac{1}{4}$ \\[1em]
\includegraphics[width=0.45\textwidth]{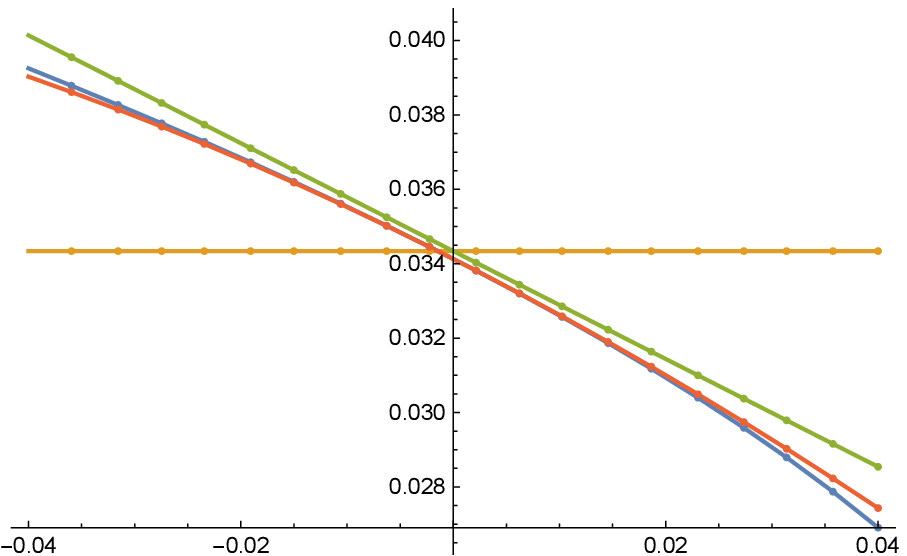}&
\includegraphics[width=0.45\textwidth]{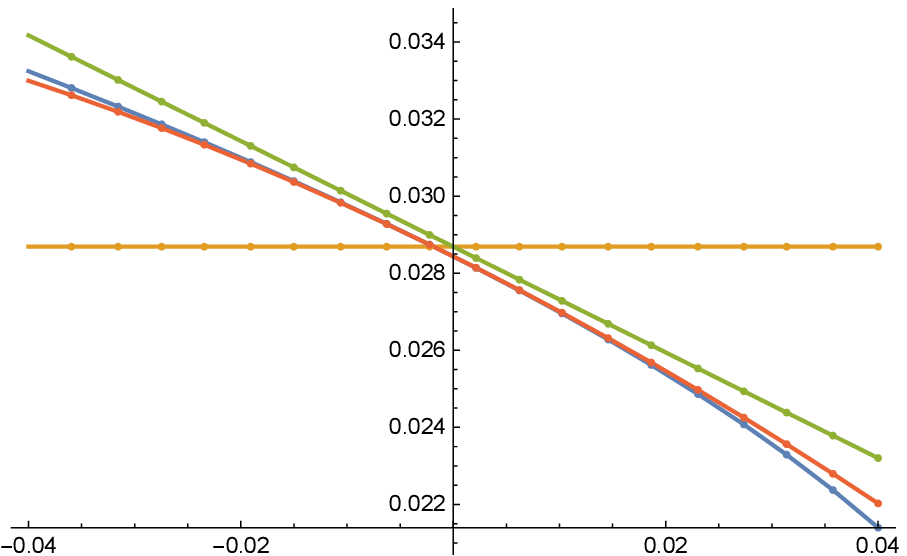}\\
$T = \frac{1}{2}$ & $T = \frac{3}{4}$ 
\end{tabular}
\caption{For the CIR short-rate model described in Section \ref{sec:cir}, we plot exact implied volatility $\Sig$ and approximate implied volatility $\Sigb_n$ up to order $n=2$ as a function of $\log$-moneyness $k-x$ with the maturity date of the bond fixed at $\Tb = 2$ and with the maturity of the option taking the following values $T = \{\frac{1}{12}, \frac{1}{4}, \frac{1}{2}, \frac{3}{4}\}$.  The zeroth, first, and second order approximate implied volatilities correspond to the {orange}, green and red curves, respectively, and the blue {curve} corresponds to the exact implied volatility.  The following parameters, which were taken from \cite[Example 10.3.2.2]{filipovic2009term}, remained fixed $t = 0$, $\kappa = 0.9$, $\delta = \sqrt{0.033}$, $\theta = \frac{0.08}{0.9}$, $y = 0.08$.}
\label{fig:CIR_IV}
\end{figure}

\begin{figure}
\centering
\begin{tabular}{c}
\includegraphics[width=0.8\textwidth]{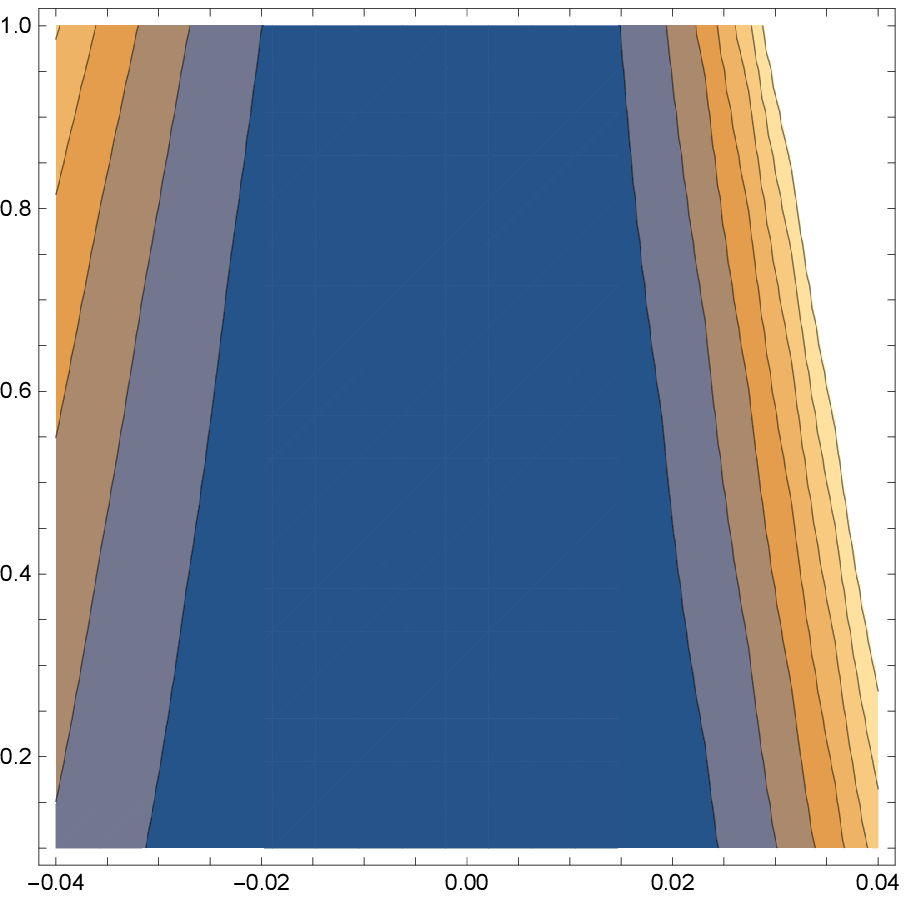}
\end{tabular}
\caption{For the CIR short-rate model described in Section \ref{sec:cir}, we plot the absolute value of the relative error of our second order implied volatility approximation $|\Sigb_2 - \Sig|/\Sig$ as a function of
log-moneyness $(k-x)$ and option maturity $T$. The horizontal axis represents log-moneyness $(k -x)$ and the vertical axis represents option maturity $T$. Ranging from
darkest to lightest, the regions above represent relative errors in increments of $0.2 \%$ from $< 0.2 \%$ to $>1.4 \%$. The maturity date of the bond is fixed at $\Tb = 2$. The following parameters, which were taken from \cite[Example 10.3.2.2]{filipovic2009term}, remained fixed $t = 0$, $\kappa = 0.9$, $\delta = \sqrt{0.033}$, $\theta = \frac{0.08}{0.9}$, $y = 0.08$.}
\label{fig:CIR-Err}
\end{figure}

\begin{figure}
\centering
\begin{tabular}{cc}
\includegraphics[width=0.45\textwidth]{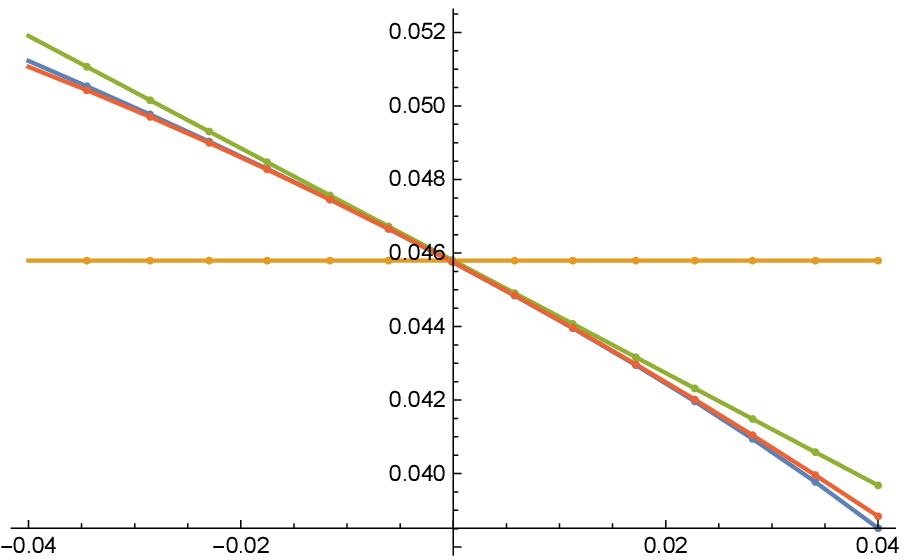}&
\includegraphics[width=0.45\textwidth]{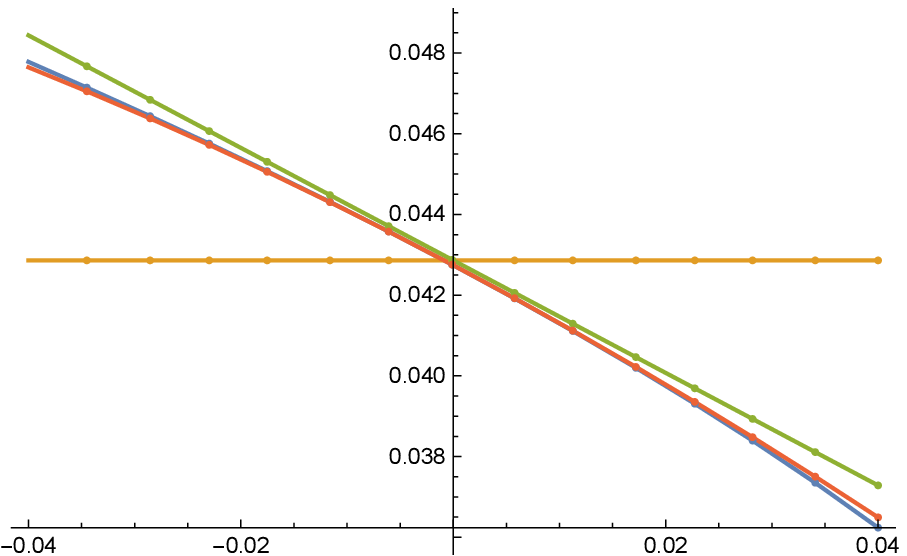}\\
$T = \frac{1}{12}$ & $T = \frac{1}{4}$ \\[1em]
\includegraphics[width=0.45\textwidth]{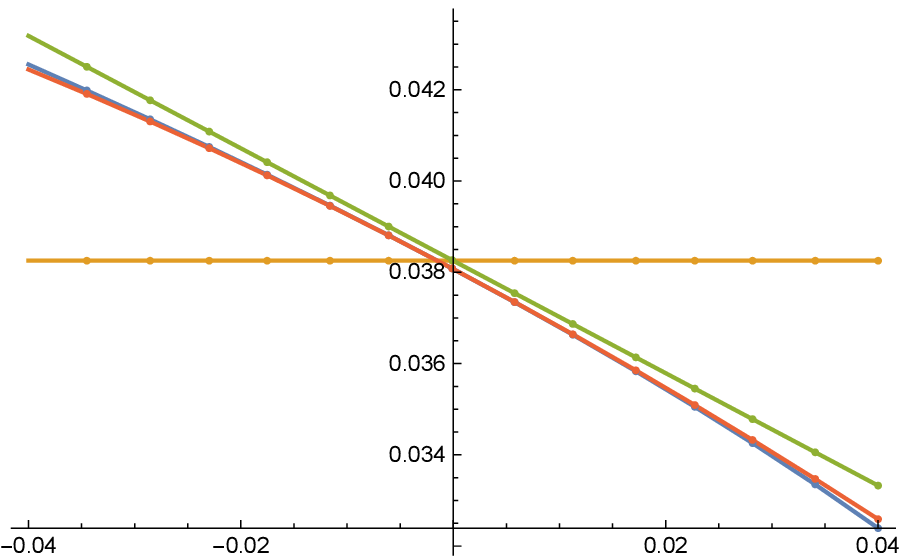}&
\includegraphics[width=0.45\textwidth]{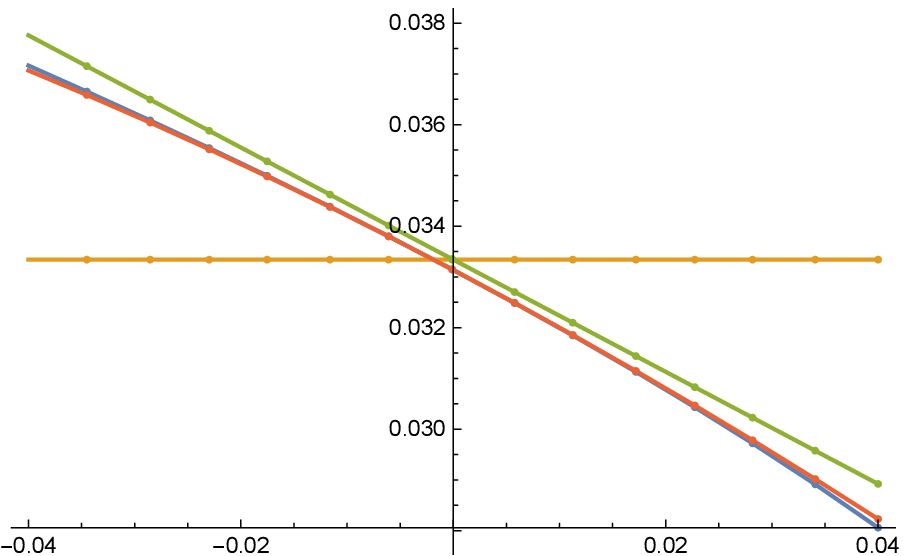}\\
$T = \frac{1}{2}$ & $T = \frac{3}{4}$ 
\end{tabular}
\caption{For the {2-D CIR} short-rate model described in Section \ref{sec:cir2}, we plot exact implied volatility $\Sig$ and approximate implied volatility $\Sigb_n$ up to order $n=2$ as a function of $\log$-moneyness $k-x$ with the maturity date of the bond fixed at $\Tb = 2$ and with the maturity of the option taking the following values $T = \{\frac{1}{12}, \frac{1}{4}, \frac{1}{2}, \frac{3}{4}\}$.  The zeroth, first, and second order approximate implied volatilities correspond to the {orange}, green and red curves, respectively, and the blue {curve} corresponds to the exact implied volatility.  The following parameters remained fixed  $t = 0$, $\kappa_1 = \kappa_2 = 0.9$, $\del_1 = \del_2 = \sqrt{0.033}$, $\theta_1 = \theta_2 =  \frac{0.08}{0.9}$, $y_1 = y_2 = 0.04$.}
\label{fig:2D-CIR-IV}
\end{figure}

\begin{figure}
\centering
\begin{tabular}{c}
\includegraphics[width=0.8\textwidth]{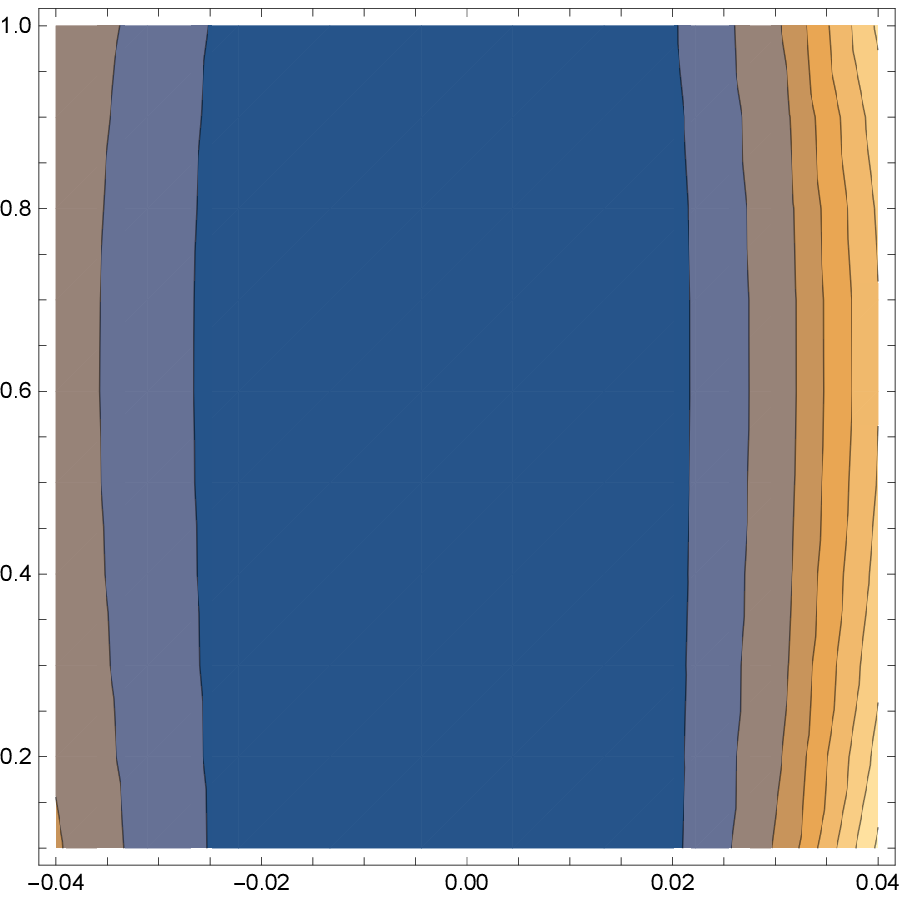}
\end{tabular}
\caption{For the {2-D CIR} short-rate model described in Section \ref{sec:cir2}, we plot the absolute value of the relative error of our second order implied volatility approximation $|\Sigb_2 - \Sig|/\Sig$ as a function of
log-moneyness $(k-x)$ and option maturity $T$. The horizontal axis represents log-moneyness $(k -x)$ and the vertical axis represents option maturity $T$. Ranging from
darkest to lightest, the regions above represent relative errors in increments of $0.1 \%$ from $< 0.1 \%$ to $>0.8 \%$. The maturity date of the bond is fixed at $\Tb = 2$. The following parameters remained fixed  $t = 0$, $\kappa_1 = \kappa_2 = 0.9$, $\del_1 = \del_2 = \sqrt{0.033}$, $\theta_1 = \theta_2 =  \frac{0.08}{0.9}$, $y_1 = y_2 = 0.04$.}
\label{fig:2D-CIR-Err}
\end{figure}

\begin{figure}
\centering
\begin{tabular}{cc}
\includegraphics[width=0.45\textwidth]{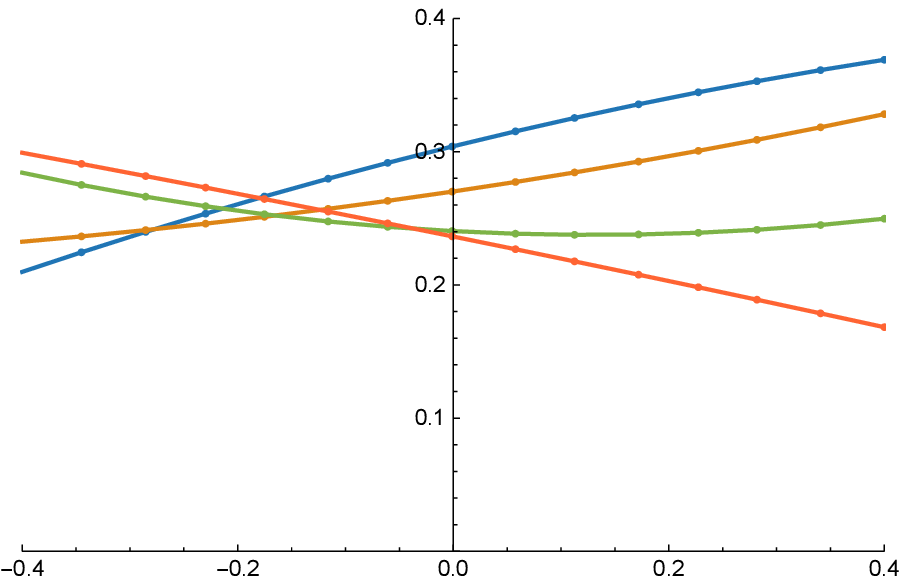}&
\includegraphics[width=0.45\textwidth]{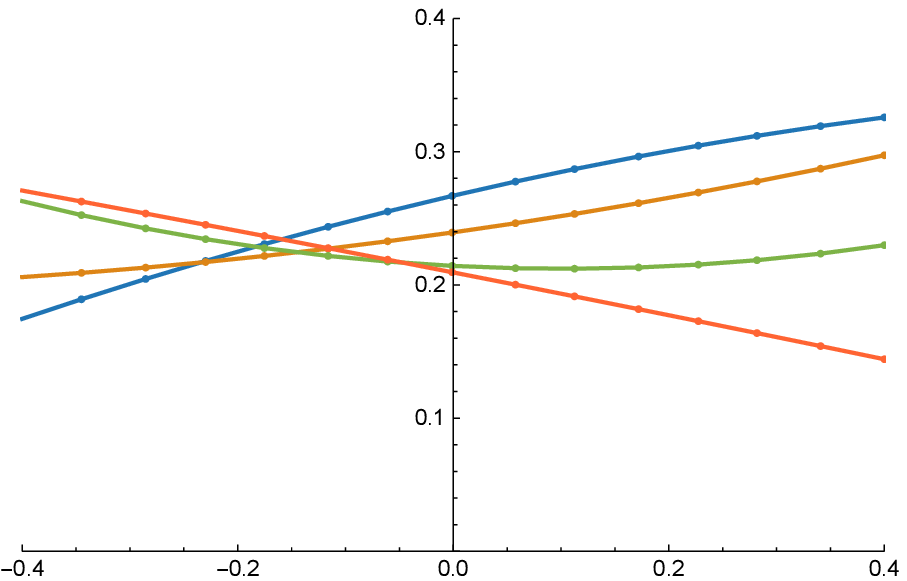}\\
$T = \frac{1}{12}$ & $T = \frac{1}{4}$ \\[1em]
\includegraphics[width=0.45\textwidth]{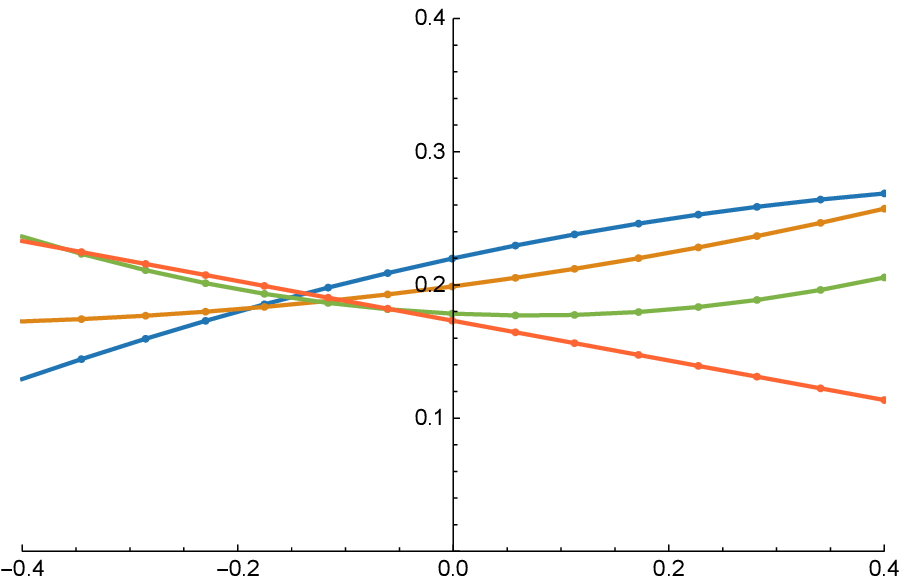}&
\includegraphics[width=0.45\textwidth]{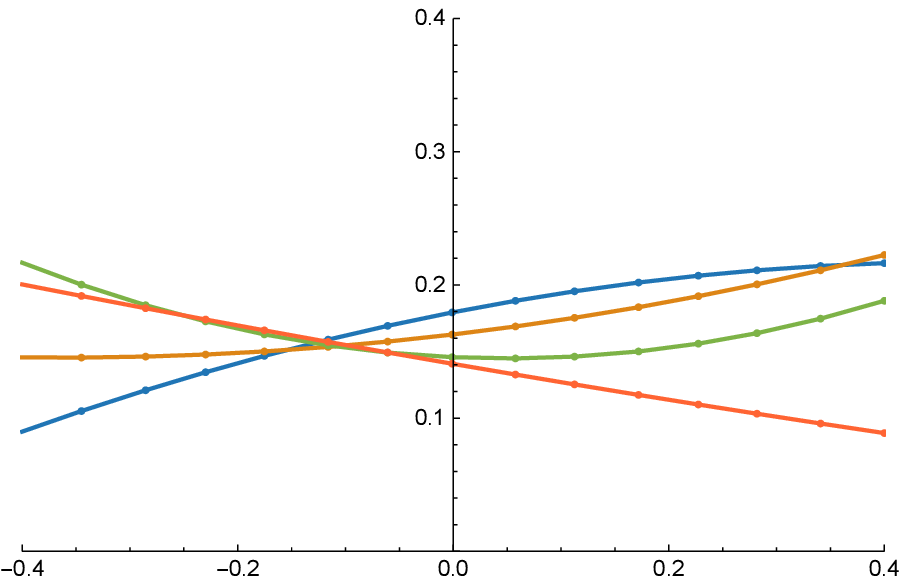}\\
$T = \frac{1}{2}$ & $T = \frac{3}{4}$ 
\end{tabular}
\caption{For the Fong-Vasicek short-rate model described in Section \ref{sec:fong}, we plot the approximate implied volatility $\Sigb_2$ as a function of $\log$-moneyness $k-x$ with the maturity date of the bond fixed at $\Tb = 2$, with the maturity of the option taking the following values $T = \{\frac{1}{12}, \frac{1}{4}, \frac{1}{2}, \frac{3}{4}\}$ and with the correlation parameter taking values $\rho = \{-0.7,-0.3,0.3,0.7\}$ corresponding to the blue, {orange}, green and red curves respectively. The following model parameters remained fixed in all four plots $t=0$, $\kappa_1 = \kappa_2 = 0.9$, $\delta_2 = \sqrt{0.08}$, $\theta_1 = \theta_2 =0.08$, $y_2 =  0.08$.}
\label{fig:plot_sigma_FV}
\end{figure}
	
\end{document}